\documentclass[letterpaper, 10 pt, conference]{ieeeconf}  
\IEEEoverridecommandlockouts                              

\usepackage{algorithm,algorithmic}
\usepackage{amsmath}
\usepackage{amsfonts}
\usepackage{amsbsy}
\usepackage{xspace}
\let\labelindent\relax 
\usepackage{enumitem}
\usepackage{graphicx}
\usepackage{indentfirst}
\usepackage{caption}
\usepackage{subcaption}
\usepackage{color}
\usepackage{multirow}
\usepackage{url}

\usepackage{array}
\newcommand{\PreserveBackslash}[1]{\let\temp=\\#1\let\\=\temp}
\newcolumntype{C}[1]{>{\PreserveBackslash\centering}p{#1}}

\DeclareMathOperator{\midd}{mid}
\DeclareMathOperator{\radd}{rad}

\DeclareMathOperator{\vol}{Vol}

\newcommand{\R}{\mathbb{R}}
\newcommand{\rd}{\delta}
\newtheorem{definition}{\textbf{Definition}}
\newtheorem{lemma}{\textbf{Lemma}}

\newtheorem{theorem}{\textbf{Theorem}}
\newcommand{\RE}{\mathbb{R}}
\newcommand{\calX}{\mathcal{X}}
\newcommand{\calB}{\mathcal{B}}
\newcommand{\calF}{\mathcal{F}}

\newcommand{\hatM}{\hat{M}}

\title{\LARGE \bf
Lagrangian Reachtubes: The Next Generation
}

\author{Sophie Gruenbacher$^{1}$, Jacek Cyranka$^{2}$, Mathias Lechner$^{3}$, Md. Ariful Islam$^{4}$, Scott A.~Smolka$^{5}$, Radu~Grosu$^{1}$
\thanks{$^{1}$Vienna University of Technology, $^{2}$Institute of Informatics, University of Warsaw, $^{3}$IST Austria, $^{4}$Texas Tech University, $^{5}$Stony Brook University}%
}

\begin{document}

\maketitle
\thispagestyle{empty}
\pagestyle{empty}

\begin{abstract}

We introduce LRT-NG, a set of techniques and an associated toolset that computes a reachtube (an over-approximation of the set of reachable states over a given time horizon) of a nonlinear dynamical system. LRT-NG significantly advances the state-of-the-art Langrangian Reachability and its associated tool LRT. From a theoretical perspective, LRT-NG is superior to LRT in three ways. First, it uses for the first time an analytically computed metric for the propagated ball which is proven to minimize the ball's volume. We emphasize that the metric computation is the centerpiece of all bloating-based techniques. Secondly, it computes the next reachset as the intersection of two balls: one based on the Cartesian metric and the other on the new metric. While the two metrics were previously considered opposing approaches, their joint use considerably tightens the reachtubes. Thirdly, it avoids the ``wrapping effect'' associated with the validated integration of the center of the reachset, by optimally absorbing the interval approximation in the radius of the next ball. From a tool-development perspective, LRT-NG is superior to LRT in two ways. First, it is a standalone tool that no longer relies on CAPD. This required the implementation of the Lohner method and a Runge-Kutta time-propagation method. Secondly, it has an improved interface, allowing the input model and initial conditions to be provided as external input files. Our experiments on a comprehensive set of benchmarks, including two Neural ODEs, demonstrates LRT-NG's superior performance compared to LRT, CAPD, and Flow*.

\end{abstract}

\section{Introduction}\label{sect:introduction}

Nonlinear ordinary differential equations (ODEs) are ubiquitous in the
modeling of cyber-physical and biological systems~\cite{FrehseLGDCRLRGDM11,grosuIcra19,DBLP:conf/hybrid/IslamMGSG14,grosuCav11,morTCS15,DBLP:journals/tac/TschaikowskiT16,cttvPNAS}. With the advent of Neural ODEs~\cite{neuralODEs}, they also became of critical importance in the Machine Learning arena, as the most popular image-recognition architectures have been shown to represent particular integration strategies of Neural ODEs~\cite{neuralODEInt}. Unfortunately, exact solutions of ODEs are limited to linear equations. The verification of nonlinear equations requires the computation of an overapproximation of their reachable states (also called a reachtube), in as tight a way as possible.

In principle, this computation is straightforward: one applies interval arithmetic to the Taylor Expansion (TE) in time of the ODEs, starting from a box (a multi-dimensional interval) representing the  set of initial states. In practice, however, this approach is too coarse~\cite{advancesValidated}, as it leads to a prohibitively large set of false positives (spurious intersections of the reachtube with the unsafe set). Hence, this approach alone is not useful in verifying cyber-physical systems (CPS).
\\[1.5mm]
\textbf{Related work.} Various techniques have been developed to tighten the above reachtube construction. Depending on how they approximate the solution set in time and space, they can be classified as follows:

\emph{TE in time, VE in space.} The variational expansion (VE) tightens the reachtubes by exploiting the deformation gradient (also known as the sensitivity) of the ODE solutions with respect to their initial states~\cite{linTE,contMec,breach,donze}. These techniques simultaneusly integrate the ODEs (to compute next reachset) and their associated variational ODEs (to compute next deformation-gradient set)~\cite{CAPD}.  State-of-the-art tools include CAPD \cite{CAPD,capdTheory,CRLOHNER} and VNode-LP \cite{vnode1,vnode2}.
These tools employ the Lohner method to reduce the wrapping effect~\cite{lohnerOrig}. As an alternative to TE the Runge-Kutta method can also be applied~\cite{Runge}.

\emph{TE in both time and space.} These techniques use Taylor Models in order to capture intervals as the sum of a symbolic part (a TE) and a remainder part (an interval). The symbolic part arises from the representation of the initial-states box in a parametric fashion, with space variables. State-of-the-art tools include COSY-Infinity~\cite{MakinoBerz2003}, the first to use such techniques with continuous systems, Flow*~\cite{DBLP:conf/cav/ChenAS13}, which extended them to hybrid systems, and CORA~\cite{Althoff2018b}. 

\emph{Bloating-based techniques.} These techniques avoid the explicit propagation of the reachset by conservatively bloating each state of an  execution (which starts from a single initial state) to a ball in some appropriate metric~\cite{MA,fansimul}. By bloating we mean increasing the diameter of the current reachset overestimate (e.g.\ a ball), with the goal of conservatively bounding all reachable states for a given time.  Discrepancy-function techniques (tool C2E2) compute the bloating by overapproximating the Jacobian of the ODEs, and compute the metric using semidefinite programming~\cite{Fan2015,Fan2016,fansimul}. Lagrangian techniques (tool LRT) compute the bloating by integrating the variational ODEs (with interval arithmetic) and the metric either analytically~\cite{CyrankaCDC18,gruenbacherArch19} or by using semidefinite programming~\cite{Cyranka2017}.

{\em Hamilton-Jacobi based techniques.} A completely different set of techniques than the ones we mentioned above are the Hamilton-Jacobi based techniques, which construct conservative reachtubes by iteratively solving the Hamilton-Jacobi partial differential equations~\cite{hjb1,hjb2}.

{\em Specific-ODEs techniques.} Some techniques are designed for very specific types of nonlinear dynamical systems, e.g., polynomial ODE systems. For these types of ODEs, the state-of-the-art approaches include, as an example, the Piecewise Barrier Tubes (PBT) approach~\cite{poly1,poly2}.
\\[1.5mm]
\textbf{Contributions.}
We significantly advance theoretical aspects of bloating-based techniques in general, and Lagrangian techniques in particular. We also considerably advance implementation aspects of LRT. We call the resulting techniques and tool \emph{LRT: The Next Generation} (LRT-NG).

The first advance concerns the metric, the centerpiece of all bloating based techniques, as everything else depends on it. Discrepancy-function techniques compute the metric in each step using semidefinite programming. The same is true for the Lagrangian technique in~\cite{Cyranka2017}. A suboptimal analytic formula was later on proposed in~\cite{CyrankaCDC18}. Hence, the above techniques lose either time or accuracy, i.e., they fail to scale up to large applications. The analytic formula proposed in this paper is proved to be indeed optimal, minimizing the volume of the reach-ball. Thus, LRT-NG has the advantage of being fast and accurate at the same time. 
The formula computes the next metric (used for the next bloating) by applying the linearization of the variational ODEs (in the next center-state) to the currently used metric. 

The second advance concerns the reachset computation. In previous work, point-state bloating was computed in either
Cartesian coordinates (with balls), or in metric-weighted coordinates (with ellipsoids). These approaches were considered opposites of one another, and were compared in terms of accuracy and speed. Due to interval arithmetic, however, the optimal ellipsoid is often not a subset of the ball, although its volume is less than the one of the ball. Consequently, we can conservatively decrease the reachset in LRT-NG by taking it as the intersection of the ball and the ellipsoid. While simple in principle, this approach considerably tightens the reachtubes, and to the best of our knowledge, is new. 

The third advance is the validated propagation of the point-states (the centers of the bloating). In previous work, a point was first conservatively propagated as a box, and subsequently solely boxes were propagated. This led to wrapping effects due to interval arithmetic. In LRT-NG, we add the size of the box to the bloating, and always propagate points, thus avoiding interval arithmetic.

In summary, our theoretical contributions resulting in much tighter reachtubes are as follows:
\begin{description}
\item{\em Metric computation.}  We introduce a new analytic method for computing the next-ball (next bloating) metric, and prove that by using this metric we minimize the volume of the resulting next-ellipsoid (Cartesian bloating).
\item{\em Reachset computation.} We considerably reduce the wrapping effect in the computation of the next reachset by intersecting the ball resulting in the Cartesian metric with the ellipsoid computed in the optimal metric.
\item{\em Center propagation.}  We show how to conservatively propagate the bloating-center states, without incurring the infamous wrapping effect due to the interval-arithmetic propagation of boxes (multi-dimensional intervals).
\end{description}

\noindent In summary, from a tool-development perspective, LRT-NG offers the following improvements:
\begin{description}
\item[\normalfont{\emph{Standalone tool.}}] LRT-NG is a standalone tool that no longer relies on CAPD to integrate the variational equations. This also required the implementation of an advanced Lohner's method, to combat the wrapping effect.
\item[\normalfont{\emph{Runge-Kutta time-integrator.}}] We base the time-stepping scheme in LRT-NG on the Runge-Kutta method~\cite{vrk2016}. This results in faster computation and, relatedly, less wrapping effect compared to the Taylor method.
\item[\normalfont{\emph{Improved interface.}}] We added a general interface to LRT-NG, allowing to input the model and initial conditions as external files. Consequently, it is now easy to run a different model or change parameters for reachtube construction, without recompiling the program. 
\end{description}

\noindent We compared LRT-NG with LRT, CAPD, and Flow* on a comprehensive set of benchmarks, including two Neural ODEs. We could not compare LRT-NG it with C2E2, as this constructs the reachtube with respect to a given set of unsafe states. In future work we plan to compare LRT-NG with CORA, too.
For each tool, we checked the tightness of the reachtube, and the time it took to construct it. Our extensive experiments show that LRT-NG is very competitive with both CAPD and Flow* on all of these benchmarks, and that it represents a major advancement of LRT. Moreover, it was the only tool able to handle the NeuralODEs. We will release LRT-NG after appropriately documenting it.

The rest of the paper is organized as follows. Section~\ref{sect:background} reviews the theoretical foundations of LRT. Section~\ref{sect:th-lrt-ng} discusses the theoretical advances in LRT-NG, while Section~\ref{sect:ta-lrt-ng} considers tool advances. Section~\ref{evaluation}, presents the results of our thorough experimental evaluation of LRT-NG and compares it to CAPD, Flow* and LRT. Section~\ref{conclusions} offers our concluding remarks and directions for future work.

\section{Background}\label{sect:background}
\emph{Notation.} We denote by $I$ the identity matrix in $\RE^{n \times n}$, and by $\partial_x$ the partial derivative with respect to variable $x$. We denote by $\|\cdot\|_2$ the \emph{Euclidean norm} with the related Cartesian Coordinates (CC).
A similar notation is used for the induced operator norms. Following standard notation, the symbol $\succ 0$ stands for positive definiteness. For a given square matrix $M_0 \succ 0$ with $A_0^TA_0\,{=}\,M_0$, let $B_{M_0}(x,\rd)$ be the closed ball with center $x$ and radius $\rd > 0$, with respect to the metric $\|x\|_{M_0}\,{=}\,\sqrt{x^T M_0 x}\,{=}\,\lVert A_0x \rVert_2$. Let $M_i\,{\succ}\,0$ with $A_i^TA_i\,{=}\,M_i$, be another square matrix. We define the metric $\lVert x\rVert_{M_{0,i}}\,{=}\,\lVert A_i x A_0^{-1} \rVert_2$.
We denote by $B(x,\rd)$ the Euclidean ball $B_I(x,\rd)$.
The notation  $[x]$ is used for a box in $\RE^n$, i.e., a product of compact intervals. We denote by $\midd([x])$ and $\radd([x])$ the midpoint and the radius of $[x]$, respectively. We call an $n\,{\times}\,n$ matrix $[\mathcal{F}]$ an interval matrix, if for all $1\le i,j \le n$, the entries $[\mathcal{F}](i,j)$ of $[\mathcal{F}]$ are intervals.

\vspace{1.5mm}We study a system of nonlinear ODEs in unknown $x\,{\in}\,\RE^n$,
where vector field $f\,{:}\,\RE^n\,{\to}\,\RE^n$ is assumed to be sufficiently smooth, time-invariant, and at least twice differentiable. As time dependence can be incorporated by adding auxiliary variable $\partial_t x_n = 1$, our discussion naturally extends to time-varying systems of the form $\partial_t x = f(t,x)$.
\begin{align}\label{cauchy}
\partial_t x & = f(x), \quad x_0 = x(t_0) 
\end{align}


Let $\chi_{t_0}^{t}(x_0)\,{=}\,x(t)$ be the solution flow induced by Eq.~\eqref{cauchy}. Suppose we would like to know how sensitive $\chi_{t_0}^{t}(x)$ is with respect to $x$. For this purpose, we introduce $F(t,x)\,{=}\,\partial_x \chi_{t_0}^{t}(x)$, called the \emph{deformation gradient} in~\cite{linTE,contMec}, and the \emph{sensitivity analysis} in~\cite{breach,donze}. $F$ varies over time according to $\partial_t F\,{=}\,\partial_t(\partial_x \chi_{t_0}^{t}(x))$. Interchanging the differentiation order, we get $\partial_t F\,{=}\,\partial_x(\partial_t \chi_{t_0}^{t}(x))$, which according to Eq.~\eqref{cauchy}, equals $\partial_t F\,{=}\,\partial_x(f(\chi_{t_0}^{t}(x)))$. By applying the chain rule, one obtains $\partial_t F\,{=}\,(\partial_{x}f)(\chi_{t_0}^{t}(x))\ \partial_{x}\chi_{t_0}^{t}(x)$.
Since $\partial_{x}\chi_{t_0}^{t}(x)\,{=}\,F(t,x)$, we get the \emph{variational equations}~\eqref{variational} associated to the system equations~\eqref{cauchy}.
\begin{align}\label{variational}
\partial_t F(t,x) =  (\partial_{x}f)(\chi_{t_0}^{t}(x))\ F(t,x), \quad F(t_0,x_0) = I
\end{align}
We define the interval matrix $[\calF_i]\supseteq\{F(t_i,x):x\in\calB_0\}$ as a compact enclosure for the gradients and call it the \emph{interval deformation gradient}. The conservative computation of $[\calF_i]$ requires the interval integration of the variational equations~\eqref{variational} starting with the initial value $[\calF_0]=I$.

LRT starts with a ball $\calB_0\,{=}\,B(x_0,\rd_0)$ in CC as the set of initial states, and computes how $\calB_i$ deforms to become $\calB_{i+1}$ as time passes. This results in a reachtube $\mathcal{B}\,{=}\,\calB_{k=0:i}$ bounding all state-trajectories of the system. Each element of
$\calB$ is an ellipsoid with $\calB_i\,{=}\,B_{M_i}(x_i, \delta_i)$, with center $x_i$ and radius $\delta_i$, where the positive definite matrix $M_i$ specifies the metric of the associated ball. 

Considering $\chi_{t_0}^{t_i}(x_0)$ as the center of the ellipsoid at time $t_i$ and after choosing an appropriate metric $M_i\succ 0$
, the task now is to find a radius $\rd_i$ which guarantees that all reachable states are conservatively enclosed by $\calB_i$.
Using the mean value inequality for vector valued functions and considering the change of metric~\cite[Lemma 2]{Cyranka2017} it holds that:
\begin{align*}
\max_{x\in\calB_0}\lVert\chi_{t_0}^{t_i}(x)&-\chi_{t_0}^{t_i}(x_0)\rVert_{M_i}\\
&\leq \max_{x\in\calB_0}\lVert F(t_i,x)\rVert_{M_{0,i}} \max_{x\in\calB_0}\lVert x-x_0\rVert_{M_0}\\
&\leq \lVert [\calF_i]\rVert_{M_{0,i}} \rd_0.
\end{align*}
Thus the challenge is to find an upper bound $\Lambda_{0,i}$ for the norm of the interval deformation gradients:
%
\begin{align}\label{eq:Lambda}
    \lVert [\calF_i]\rVert_{M_{0,i}}\leq \Lambda_{0,i} \Rightarrow \rd_i = \Lambda_{0,i}\rd_0.
\end{align}
We call $\Lambda_{0,i}$ the \emph{stretching factor (SF)} associated to the interval gradient tensor.
Having the interval gradient $[\calF_i]$ at time $t_i$ we solve equation~\eqref{eq:Lambda} using algorithms from~\cite{hladik,rump,rohn}, and choosing the tightest result available.
%
%
%

%
The correctness of LRT is rooted in~\cite[Theorem 1]{Cyranka2017}. 
\section{Theoretical Advances in LRT-NG}\label{sect:th-lrt-ng}

This section presents the theoretical advances of LRT-NG. In particular, Section~\ref{sec:metric} first states the optimization problem to be solved in order to get the optimal metric, and thus the reachset with minimal volume. We then introduce an analytic solution, prove that it solves the optimization problem, and discuss the intuition behind it. Section~\ref{sec:reacxh} focuses on the new reachset computation, as the intersection of an ellipsoid and a ball, and illustrates the results of this computation. Finally, Section~\ref{sec:center} shows how to conservatively propagate the center without incurring the infamous wrapping effect of interval arithmetic.

\subsection{Computation of the Metric}\label{sec:metric}

To obtain an as-tight-as-possible over-approximation, we wish to minimize the volume of the n-dimensional ball $\calB_i$, that is of $B_{M_i}(\chi_{t_0}^{t_i}(x_0),\delta_i)$. This volume is given by~\cite{volEllipsoid}, recalled below, where $a_j$ denotes the length of the semi-axis $j$ of the ellipsoid associated to $\calB_i$ in CC ($a_j = \delta_i \, \lambda_j(M_i)^{-1/2}$ with $\lambda_j(M_i)$ being the $j$-th eigenvalue of $M_i$):
%
\begin{align}\label{vol1}
\vol(\calB_i) = \frac2{n} \frac{\pi^{\frac{n}{2}}}{\Gamma(\frac{n}{2})} \prod_{j=1}^n a_j,
\end{align}
where $\Gamma$ is the Gamma-function. It holds that $\Lambda_{0,i}\,(M_i)\,{\geq}\,\lVert[\calF_i]\rVert_{M_{0,i}}$. Given that the set $[\calF_i]$ is represented by an interval matrix in our algorithm, we pick the $M_i (=A_i^T A_i)$ that minimizes the ellipsoid volume with respect to the gradient $F_i\in [\calF_i]$ in the midpoint $x_0$. Thus the optimization problem is given by ($\delta_i\,{=}\,\Lambda_{0,i}\,(M_i)\,\delta_{0}$):
\begin{align}\label{vol2}
&\min_{M_i\succ 0}\vol\left(B_{M_i}(\chi_{t_0}^{t_i}(x_0),\Lambda_{0,i}(M_i) \delta_0)\right)=\{\mathrm{\eqref{vol1},\ def}\ \calB_i\}\nonumber\\
&\min_{M_i\succ 0}\frac2{n} \frac{\pi^{\frac{n}{2}}}{\Gamma(\frac{n}{2})}\rd_i^n \prod_{j=1}^n \lambda_j(M_i)^{-1/2}=\{\mathrm{\tiny outsource\ Cst:\ C(n)}\}\nonumber\\
&C(n)\min_{M_i\succ 0} \Lambda_{0,i}(M_i)^n \,\det(M_i)^{-1/2}
\end{align}
Let us further define $\hat{F}_{i-1,i} = \partial_x\chi_{t_{i-1}}^{t_i}(x)|_{x=\chi_{t_0}^{t_{i-1}}(x_0)}$ as the deformation gradient from time $t_{i-1}$ to $t_i$ at the center of the ball. Using the chain rule it holds that $F_i = \prod_{j=1}^i \hat{F}_{j-1,j}$, where $F_i$ is defined as the deformation gradient at $x_0$.
\vspace{2ex}
\begin{definition}[New metric $\hatM_i\,{\succ}\,0$]\label{def:M1}
\textit{Let the gradient-of-the-flow matrices $F_i$ and $\hat{F}_{i-1,i}\,{\in}\,\R^{n\times n}$ be full rank, and the coordinate-system matrix of the last time-step $A_{i-1}\,{\in}\,\R^{n\times n}$ be full-rank and $A_{i-1}\,{\succ}\,0$. Define metric $\hat{M}_i(F_{i})=\hat{A}_i(F_i)^T\hat{A}_i(F_i)$ with:}
\begin{align}\label{M1opt}
    \hat{A}_i(F_i) = A_{i-1}\hat{F}_{i-1,i}^{-1} = A_0F_i^{-1} 
\end{align}
\textit{When $F_i$ is known, we simply abbreviate $\hat{A}_i(F_i)$ with $\hat{A}_i$, and $\hatM_i(F_i)$ with $\hatM_i$.}
\end{definition}
\vspace{2ex}
In order to find an optimal metric $M_i$, the theoretical approach of LRT presented in~\cite{Cyranka2017,CyrankaCDC18,gruenbacherArch19} solved the following optimization problem:
\begin{align}\label{vol3}
\min_{M_i\succ 0}\lVert[\calF_i]\rVert_{M_{i}}
\end{align}
This is different from~\eqref{vol2}, and its goal was to find the metric $\tilde{M}_i$ minimizing the $M_i$-SF. In~\cite{Cyranka2017}, $M_i$ was found through nonlinear optimization, which was time consuming and sub-optimal. In~\cite{CyrankaCDC18,gruenbacherArch19},  this was improved, by developing a simple explicit analytic formula for $\tilde{M}_i$. But in fact, this was a correct result for the wrong optimization problem.

In this paper, we do not minimize the maximum singular value of the interval gradient, as in~\eqref{vol3}, but we directly minimize the volume of the n-dimensional ball $\calB_i$, as in~\eqref{vol2}, at each time-step.
The new analytical formula for the optimal metric $\hatM_i$ in Def.~\ref{def:M1}, results in considerable tighter bounds, compared to the old $\tilde{M}_i$ in~\cite[Definition 1]{CyrankaCDC18}, as we prove below in Theorem~\ref{thmopt}. 

The new metric takes into account the previous metric $M_{i-1}$, and explicitly minimizes the volume of the resulting ellipsoidal bounds by using a smart choice of $M_i$. Intuitively we find a metric $\hatM_i$ which describes the shape of the set of reachable states at a certain time, which results in the problem of minimizing the volume of the n-dimensional ellipsoid as described in~\eqref{vol2}. The metric $\hatM_i$ optimizing this problem is the one in Definition~\ref{def:M1}: $\hat{A}_i^{-1}\,{=}\, \hat{F}_{i-1,i}\,A_{i-1}^{-1}$. Intuitively, we transform the current shape represented by the coordinate system $A_{i-1}^{-1}$ using $\hat{F}_{i-1,i}$, the linearization of the system dynamics. Thus with $\hat{M}_i$, we actually describe the current shape of the set of reachable states. A comparison of the reachsets obtained with LRT metric $\tilde{M}_i$ and the one obtained with LRT-NG metric $\hat{M}_i$ is illustrated in Fig.~\ref{fig:metricComparison}. It shows that the LRT metric is by far not an optimal choice, and it also shows how well our new analytically computed metric $\hat{M}_i$ follows the shape of the set of reachable states.

\begin{figure}[t]
    \centering
    \includegraphics[width=0.45\textwidth]{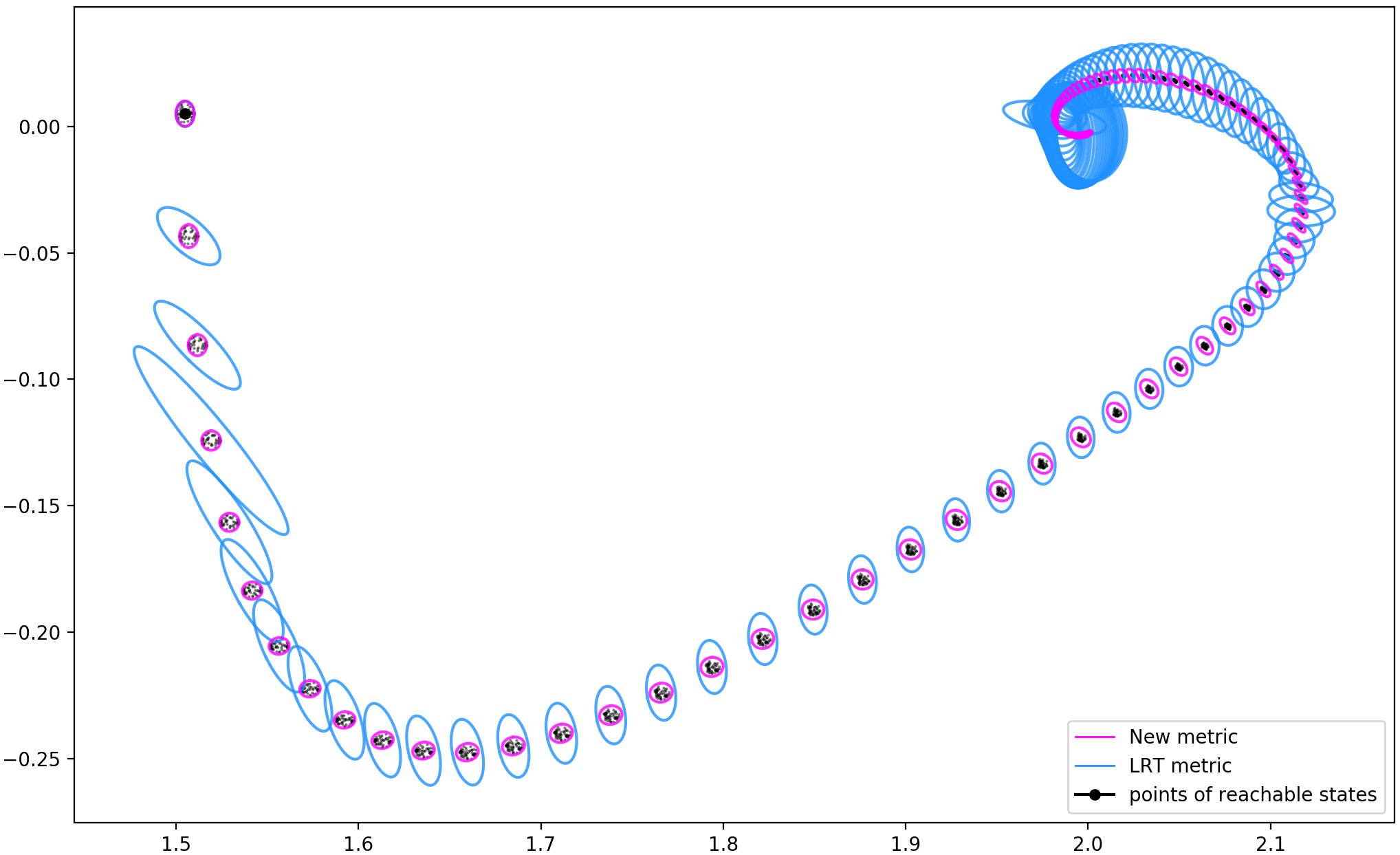}
    \vspace{-1.ex}
    \caption{Reachtube for Robotarm model ($t\in[0,5]$) obtained with LRT and new metric.}
    \label{fig:metricComparison}
\vspace{-3ex}
\end{figure}

The following theorem states that metric $\hat{M}_i$ minimizes the ellipsoid volume in~\eqref{vol2}, and it is therefore optimal.
\vspace{2ex}
\begin{theorem}\label{thmopt}
\textit{Let $F_i \in \mathbb{R}^{n\times n}$ be a gradient having full rank, and $\hat{A}_i$, $\hat{M}_i$ be defined by \eqref{M1opt}. Let the $M_{0,i}$-SF be given by (with $M_0$ fixed):
\[
\Lambda_{0,i}(M_i)=\sqrt{\lambda_{\max}\left((A_0^T)^{-1}F_i^TM_iF_iA_0^{-1}\right)}.
\]
Then, it holds that $\ \vol\left(B_{\hat{M_i}}(\chi_{t_0}^{t_i}(x_0),\Lambda_{0,i}(\hat{M_i})\,\delta_0)\right)$ is equal to:
\[
    \min_{M_i\succ 0} \vol\left(B_{M_i}(\chi_{t_0}^{t_i}(x_0),\Lambda_{0,i}(M_i)\,\delta_0)\right).
\]
In other words, the symmetric matrix $\hatM_i\,{\succ}\,0$ minimizes the volume of the ellipsoid $B_{M_i}(\chi_{t_0}^{t_i}(x_0),\Lambda_{0,i}(M_i)\,\delta_0)$ as a function of $M_i$.}
\end{theorem}
\vspace{2ex}
\begin{proof}
For an arbitrary matrix $D\,{\in}\,\R^{n\times n},\, D\,{\succ}\,0$ it is straightforward to show that $\lambda_{\max}(D)\,{\geq}\,\det(D)^{\frac1{n}}$ holds. Hence:
\vspace{-3ex}
\[
\Lambda_{0,i}(M_i)\geq
\lVert F_i\rVert_{M_{0,i}}
\geq \left(\frac{\det(F_i)^2\, \det(M_i)}{\det(A_0)^2}\right)^{\frac1{2n}}.
\]
Using the volume of an ellipsoid given by Eq.~\eqref{vol1} it follows:
\vspace{-1ex}
\[
\vol(\calB_i) =
C(n)\,\Lambda_{0,i}(M_i)^n \,\det(M_i)^{-\frac1{2}} \geq C(n)\,\frac{\det(F_i)}{\det(A_0)},
\vspace{-1ex}
\]
%
%
We show that this lower bound is attained for $\hatM_i$. By Definition~\ref{def:M1}, $\hat{A}_i = A_0F_i^{-1}$ and thus $\Lambda_{0,i}(\hatM_i)$ is equal to:
\[
\lVert F_i\rVert_{M_{0,i}}=\lVert A_0F_i^{-1}F_iA_0^{-1}\rVert_2
=1
\]
Hence, the volume $\vol(\calB_i,\hatM_i)$
is equal to:
\[
C(n)\,\det(\hatM_i)^{-1/2} =
C(n)\,\frac{\det(F_i)}{\det(A_0)}
\vspace{-2ex}
\]
\end{proof}

Thus, Theorem~\ref{thmopt} gives us an analytic solution for the optimal metric, releasing us from either solving an optimization problem with semi-definite programming in every time-step like in~\cite{fansimul,Cyranka2017}, or risking false-positive consequences of using a suboptimal metric as in LRT~\cite{CyrankaCDC18,gruenbacherArch19}.

\subsection{Intersection of the Reachsets}\label{sec:reacxh}

Another novelty in LRT-NG, is that it defines the next reachset as the intersection of an ellipsoid computed in the Optimal metric and a ball computed in the Euclidean metric. This considerably reduces the volume and therefore increases the time horizon of the computed reachtube.

The integration of the interval deformation gradient $[\calF_{i+1}]$ is based on the interval enclosure $[\calX_i]$ of the current reachset $B_{M_i}\left(\chi_{t_0}^{t_i}(x_0),\delta_i\right)\subseteq [\calX_i]$, such that the tighter $[\calX_i]$, the tighter the next deformation gradient $[\calF_{i+1}]$ is. When computing the SF $\Lambda_{0,i+1}$ for time $t_{i+1}$, the singular value of $A_{i+1}F_{i+1}A_0^{-1}$ is maximized over all $F_{i+1}\in[\calF_{i+1}]$. Hence it is crucial that $[\calF_{i+1}]$ and thus also $[\calX_i]$ is as tight as possible to decrease not only the volume of the reachset at time $t_{i+1}$, but also the volume of the reachsets at all times $t\ge t_{i+1}$.
%

An effective way of getting a much tighter conservative bound $[\calX_i]$ is taking the intersection of the ellipsoid in the optimal metric $\hatM_i$, and the ball in Euclidean metric.

As one can clearly see in Fig.~\ref{fig:intersectionPlot}, this is a considerable improvement. The blue ellipsoids are already in the optimal metric $\hatM_i$ (representing the real shape of the set of reachable states), but we can still get much tighter reachsets by using the intersections technique presented in this chapter, when computing the interval enclosure $[\calX_i]$. That new approach is conservative, as proven in the following proposition.
\begin{figure}[t]
    \centering
    \includegraphics[width=0.45\textwidth]{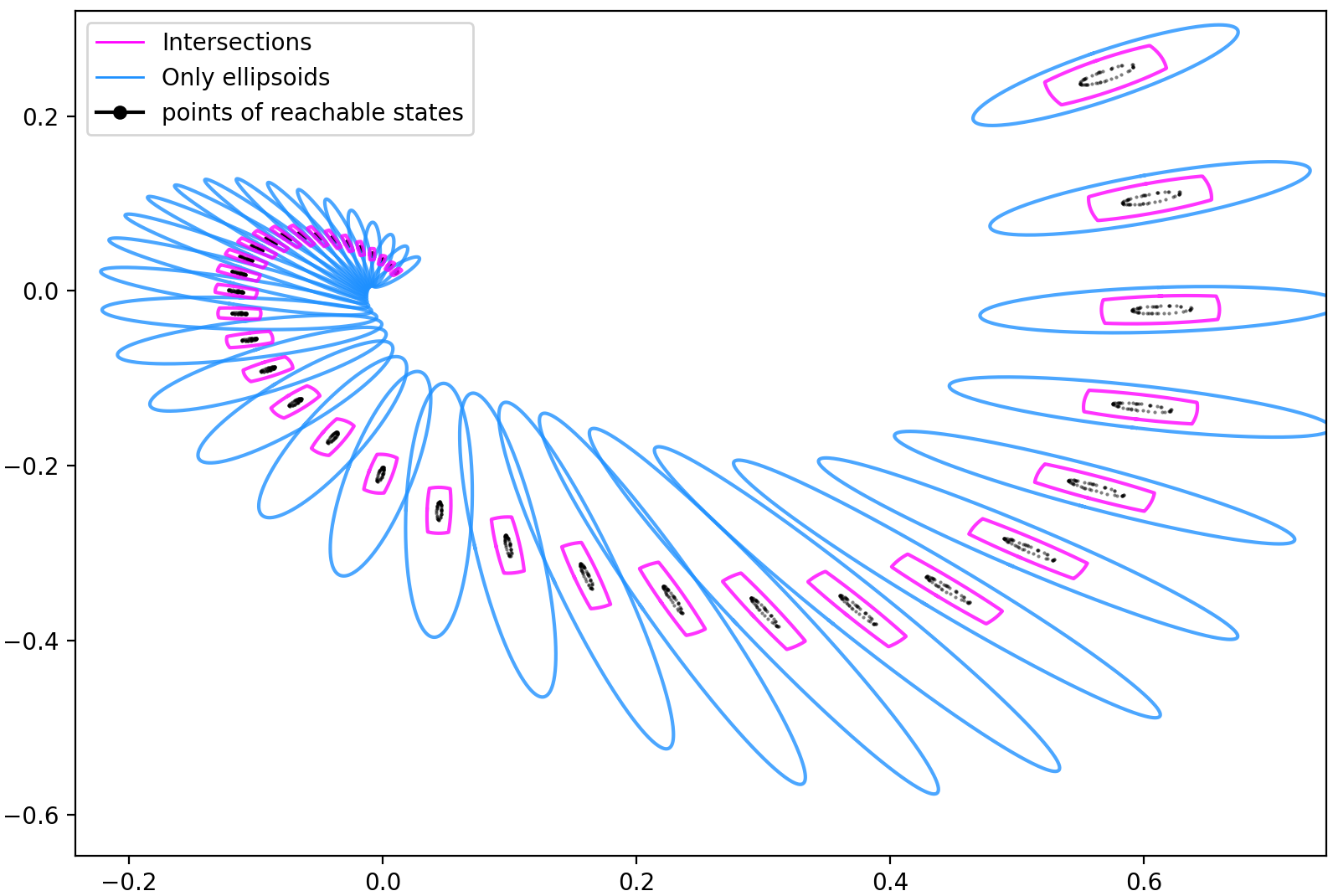}
    \vspace{-1.ex}
    \caption{Reachtubes for Van der Pol Oscillator ($t\in[3.4,5]$) for different $[\calX_i]$.}
    \label{fig:intersectionPlot}
\vspace{-3ex}
\end{figure}
\vspace{2ex}
\begin{lemma}\label{lemma:intersection}
\textit{Let ball $\calB_0\,{=}\,B_{M_0}(x_0,\delta_0)\subseteq\mathbb{R}^n$ in $M_0$-norm, with center $x_0$, and radius $\delta_0$, be the set of initial states, at time $t_0$. Let $\Lambda_{0,i}$ be the upper bound of the $M_{0,i}$-SF, and $\Lambda_{i,M_0}$ be the $M_0$-SF at time $t_i$. Let the radius $\delta_i$ be defined as in Eq.~\ref{eq:Lambda}, and $\delta_{i,M_0}$ as below:
\begin{align*}
    \delta_{i,M_0} &=  \Lambda_{i,M_0}\,\delta_0.
\end{align*}
Then, it holds that:
\begin{align*}
    \chi_{t_0}^{t_i}(x)&\in \calB_i\cap B_{M_0}\left(\chi_{t_0}^{t_i}(x_0),\delta_{i,M_0}\right),\quad \forall x\in\calB_0.
\end{align*}
The set of reachable states are conservatively enclosed by the intersection of the balls in the two different metrics.}
\end{lemma}
\begin{proof}
The conservativity of the LRT-NG is independent of the choice of metric. Therefore we can choose $\hat{M}_i\,{=}\,M_0$ and $\tilde{M}_i\,{=}\,M_i$ to be two different metrics at time $t_i$. Using Theorem~\cite[Theorem 1]{Cyranka2017}, we know that $\chi_{t_0}^{t_i}(x)\,{\in}\,B_{\hat{M}_i}(\chi_{t_0}^{t_i}(x_0),\rd_i(\hat{M}_i))$ and that $\chi_{t_0}^{t_i}(x)\,{\in}\,B_{\tilde{M}_i}(\chi_{t_0}^{t_i}(x_0),\rd_i(\tilde{M}_i))$. Hence $\chi_{t_0}^{t_i}(x)$ is also in the intersection of the ellipsoid and the ball, for all $x\,{\in}\,\calB_0$.
\end{proof}
\vspace{2ex}
Thus Lemma~\ref{lemma:intersection} allows us to dramatically reduce the volume of the reachtube and combat the wrapping effect in a way that has not been considered before by bloating-based techniques~\cite{Cyranka2017,CyrankaCDC18,gruenbacherArch19,fansimul,Fan2016,Fan2015}.
\subsection{Propagation of the Center}\label{sec:center}

In this section we show how to conservatively propagate the center and avoid the wrapping effect incurred due to interval arithmetic over boxes. The main idea is to add the box size to the radius of the ellipsoid and propagate points.

As ball $\calB_i\,{=}\,B_{M_i}(\chi_{t_0}^{t_i}(x_0),\rd_i)$ at time $t_i$ in Theorem~\cite[Theorem 1]{Cyranka2017}, we need first to propagate $x_0$ to $\chi_{t_0}^{t_i}(x_0)$ in a validated fashion. This allows us to then describe this conservative reachset. In LRT~\cite{Cyranka2017,CyrankaCDC18,gruenbacherArch19} $x_0$ was propagated in the first time-step to $[y_1]$, a validated interval including the solution $\chi_{t_0}^{t_1}(x_0)$. After that, the interval $[y_i]$ which includes the solution $\chi_{t_0}^{t_i}(x_0)$, was propagated to $[y_{i+1}]$ from time $i$ to time $i\,{+}\,1$. This way, it was possible to define the ellipsoidal bound by using that interval as the center of ball $B_{M_i}([y_i],\rd_i)$, but with a loss in tightness: it accumulated the wrapping effect due to interval arithmetic.

In this paper we avoid the wrapping effect of the center propagation in interval arithmetic by restarting the center to the midpoint $x_i=\midd{[y_i]}$ in each time-step and adding an upper bound of the error $\lVert \chi_{t_0}^{t_i}(x_0)\,{-}\, x_i\rVert_{M_i}\,{\leq}\,\sigma_{i,M_i}$ to the radius $\rd_i$ of the ellipsoid. Using the new definition of the radius $\rd_i\,{=}\, \Lambda_{0,i}\,\rd_0\,{+}\,\sigma_{i,M_i}$, we are able to define the ball enclosing the reachable states at time $t_i$ as follows:
$$\calB_i = B_{M_i}(x_i,\rd_i),$$
This allows us to use in each time-step a point $x_i$ as the center of the ball, instead of using an interval $[y_i]$.
In order to properly use the midpoint of the propagated center as the new center of the ball, we need to calculate the maximum error in each step. The following Lemma provides a formula for calculating the error $\sigma_{i,M_i}$, with $\lVert \chi_{t_0}^{t_i}(x_0)\,{-}\,x_i\rVert_{M_i}\,{\leq}\,\sigma_{i,M_i}$.
\vspace{2ex}
\begin{lemma}\label{lemma:sigma}
\textit{Let $x_0$ be the center of the set of initial states at the initial time $t_0$, and let the interval $[y_i]$ be a validated interval bound for the solution set of Eq.~\eqref{cauchy} with initial condition $(t_{i-1}, x_{i-1})$, thus $\chi_{t_{i-1}}^{t_i}(x_{i-1})\,{\subseteq}\,[y_{i}]$. Let $x_i\,{=}\,\midd([y_i])$ be the middle of the validated interval bound, and $\epsilon_{i, M_i}\,{=}\,\radd(A_i\,[y_i])$ the radius of the interval $[y_i]$ in metric $M_i$, corresponding to the maximum propagation error of the center in every dimension.
Assume there exists a compact enclosure $[\mathfrak{C}]\subseteq\mathbb{R}^{n\times n}$ for the gradients in the $\sigma$-area around the center such that:
$$(\partial_x \chi_{t_{i-1}}^{t_{i}})(x)\in [\mathfrak{C}], \quad \forall x\in B_{M_{i-1}}(x_{i-1},\sigma_{i-1,M_{i-1}}).$$
Let $\Lambda_{i-1,i}$ be the following upper bound:
\begin{align*}
    \Lambda_{i-1,i} \geq \lVert C \rVert_{M_{i-1,i}},\quad  \forall C\in[\mathfrak{C}],
\end{align*}
which is the stretching factor from time $t_{i-1}$ to time $t_i$ only in this small region around the center within the error radius $\sigma_{i-1,M_{i-1}}$.
Starting with $\sigma_{0,M_0}\,{=}\,0$, let the error at time $t_i$ be defined as follows:
\begin{align}\label{eqn:defsigma}
    \sigma_{i,M_i} = \Lambda_{i-1,i}\,\sigma_{i-1,M_{i-1}} + \epsilon_{i,M_i}
\end{align}
Then it holds that: 
\begin{align}\label{eqn:sigma}
    \chi_{t_0}^{t_i}(x_0)&\in B_{M_i}(x_i,\sigma_{i,M_i}).
\end{align}
}

\end{lemma}
\begin{proof}
We prove by induction, that this holds for all $i\in\mathbb{N}$. \emph{For time $t_1$}:
\begin{align*}
    \lVert \chi_{t_0}^{t_1}(x_0) - x_1\rVert_{M_1} &= \lVert A_1 \, (\chi_{t_0}^{t_1}(x_0) - x_1)\rVert_2 \leq\\
    &\max_{x\in[y_1]}\lVert A_1\,(x - x_1)\rVert_2 = \epsilon_{1,M_1} = \sigma_{1,M_1},
\end{align*}
which is equivalent to $\chi_{t_0}^{t_1}(x_0)\,{\in}\,B_{M_1}(x_1,\sigma_{1,M_1})$. \emph{Now assuming~\eqref{eqn:sigma} holds for $t_i$}. We show that $\chi_{t_0}^{t_{i+1}}(x_0)\,{\in}\,B_{M_{i+1}}(x_{i+1},\sigma_{i+1,M_{i+1}})$ holds for $t_{i+1}$:
\begin{align*}
    &\lVert \chi_{t_0}^{t_{i+1}}(x_0) - x_{i+1}\rVert_{M_{i+1}}\leq \mathrm{\{triangle\ ineq\}}\\
    &\lVert \chi_{t_0}^{t_{i+1}}(x_0) - \chi_{t_i}^{t_{i+1}}(x_i)\rVert_{M_{i+1}} + \lVert \chi_{t_i}^{t_{i+1}}(x_i) - x_{i+1}\rVert_{M_{i+1}}\leq\\
    &\lVert \chi_{t_i}^{t_{i+1}}\left(\chi_{t_0}^{t_{i}}(x_0)\right) - \chi_{t_i}^{t_{i+1}}(x_i)\rVert_{M_{i+1}} + \epsilon_{i+1,M_{i+1}}\leq\\
    &\max_{x\in \calB_i}\lVert(\partial_x \chi_{t_i}^{t_{i+1}})(x)\rVert_{M_{i,i+1}}\lVert \chi_{t_0}^{t_{i}}(x_0) - x_i \rVert_{M_i} + \epsilon_{i+1,M_{i+1}}\\
    &\leq \mathrm{\{ind\ hyp\}}\ \Lambda_{i,i+1}\, \sigma_{i,M_i} + \epsilon_{i+1,M_{i+1}}=\sigma_{i+1,M_{i+1}}.
\vspace{-2ex}
\end{align*}
\end{proof}
Following theorem proves the correctness of LRT-NG. The center of the ball is the midpoint $x_i$ of the validated interval bound of the propagation $\chi_{t_{i-1}}^{t_i}(x_{i-1})$. This is close to the real center $\chi_{t_0}^{t_i}(x_0)$, and adds the error to the radius.
\vspace{2ex}
\begin{theorem}\label{thmrad}
\textit{Let the ball $\calB_0\,{=}\,B_{M_0}(x_0,\delta_0)\subseteq\mathbb{R}^n$, be the set of initial states at $t_0$, $\Lambda_{0,i}$ the upper bound of $M_{0,i}$-SF, and the new radius $\delta_i$ be defined as:
\begin{align*}
\delta_i &= \Lambda_{0,i}\, \delta_0 + \sigma_{i,M_i},
\end{align*}
with $\sigma_{i,M_i}$ as defined in Eq.~\eqref{eqn:defsigma}.
Then, the propagation error of the center:
\begin{align}\label{eqn:radius}
    \chi_{t_0}^{t_i}(x)&\in B_{M_i}\left(x_i,\delta_i\right)=\calB_i,\quad \forall x\in\calB_0
\end{align}
can be absorbed in the radius $\delta_i$ of each time-step.}
\end{theorem}
\begin{proof}
By triangle inequality $\lVert \chi_{t_0}^{t_{i}}(x)-x_{i}\rVert_{M_{0,i}}\leq \lVert \chi_{t_0}^{t_{i}}(x) - \chi_{t_0}^{t_{i}}(x_0)\rVert_{M_{0,i}} + \lVert \chi_{t_0}^{t_{i}}(x_0) - x_{i}\rVert_{M_{0,i}}$
By Theorem~\cite[Theorem 1]{Cyranka2017}, the first term is bounded by $\Lambda_{0,i}\,\delta_0$, and by Lemma~\ref{lemma:sigma}, the second is bounded by $\sigma_{i,M_{i}}$. Thus:
$\chi_{t_0}^{t_i}(x)\in B_{M_i}\left(x_i,\delta_i\right)=\calB_i,\ \forall x\in\calB_0$.
\end{proof}
\vspace{2ex}
\begin{figure}[t]
    \centering
    \includegraphics[scale=0.52]{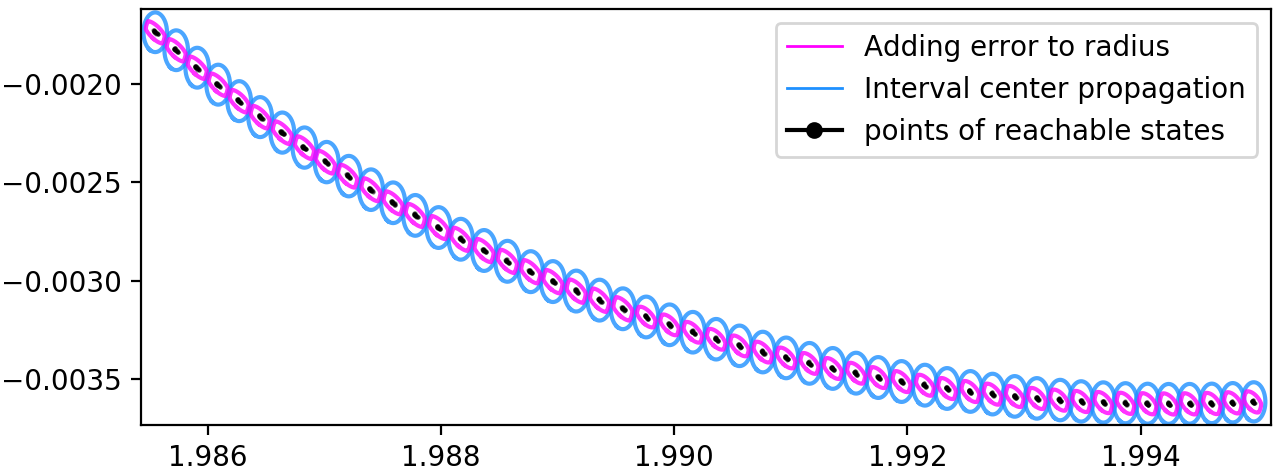}
    \vspace{-1.ex}
    \caption{Reachtube for Robotarm model ($t\in[9,10]$) obtained with the LRT (blue) and the LRT-NG (pink) method for center propagation.}
    \label{fig:centerPlot}
\vspace{-3ex}
\end{figure}

As one can see in Fig.~\ref{fig:centerPlot}, this method avoids the wrapping effect of the center propagation in interval arithmetic by restarting the center to the midpoint in each time-step and adding an upper bound of the error to the radius.

\section{Tool Advances in LRT-NG}\label{sect:ta-lrt-ng}
LRT-NG also entails several improvements from a tool-development point of view as compared to the previous implementation of LRT, presented in \cite{CyrankaCDC18}. 

\subsection{Runge-Kutta time integrator using Lohner's method}\label{sec:integrator}

To compute the stretching factor $\Lambda_{0,i}$ needed for Theorem~\ref{thmrad}, we have to conservatively compute a compact enclosure $[\calF_i]$ of the deformation gradients $F(t_i,x)$ for all $x\,{\in}\,\calB_0$. This requires the interval integration of the variational equations~\eqref{variational}. In the work on LRT~\cite{Cyranka2017,CyrankaCDC18}, the external CAPD library was used for this purpose, where the algorithm for time-propagation is the Taylor method. In this work, the LRT-NG tool is standalone.

For scalability reasons, we implemented in LRT-NG from scratch a rigorous procedure for time-propagation of the gradient using the Runge-Kutta (RK) time-stepping scheme \cite{Runge,Kutta,butcher}. Authors of \cite{vrk2016} also employ RK with the same motivation as ours, but they utilize zonotopes for wrapping control. We believe that our new approach, results in much lower computational complexity, while maintaining high accuracy. At present, we support only first, second, and fourth order RK methods, as higher orders while slower, considerably suffer from wrapping effects. Thus, they do not really lead to more improvements. 

\subsection{New Tool-Interface and Implementation}\label{sec:interface}

Other tool-development improvements concern the tool itself, combating wrapping effect, and interface details.

LRT-NG tool is  standalone. The only external libraries that we require are the IBEX and the Eigen C++ library \cite{ibex,eigen}. These libraries provide a convenient implementation of interval arithmetic, basic data structures, including vectors and matrices, eigenvalue computation, and symbolic differentiation routines, which we use for computing Jacobians.

To combat wrapping effects in the integration of~\eqref{variational} with RK, we implemented from scratch an advanced Lohner's QR method~\cite{Lo,lohnerOrig,capdTheory,CRLOHNER,gruenbacherArch19}. This required the efficient propagation of the gradients, and the removal of the dependencies on any non-fundamental external libraries.

Regarding the LRT-NG interface, to construct a reachtube for a new benchmark the user only needs to provide an external file with a function implementing the differential equations of the studied model and a file with the initial set of states for the reachtube construction. We believe this approach is the most flexible. First, it allows for fast and convenient reachtube construction using different models and parameters like time-step, time-horizon, RK order, without recompiling the program. Secondly, if the user wants to perform a verification of a large system of equations, then usually a simulation code is implemented first. Then, the script performing the simulation can be directly translated to the input to LRT-NG.

\section{Experimental Evaluation}\label{evaluation}

To assess the performance of LRT-NG in terms of accuracy and speed, we applied it to a comprehensive set of available nonlinear ODE benchmarks, and to two Neural ODEs we developed ourselves. The benchmarks are as follows: Brusselator, Van der Pol Oscillator, Robotarm, Dubins Car, Mitchell-Schaeffer Cardiac-cell, linearly controlled Cartpole, Quadcopter, Cartpole controlled with a Neural ODE~\cite{neuralODEInt}, and a Cartpole controlled with an LTC~\cite{ltc}. For comparison, we applied LRT and the latest versions of Flow* and CAPD to these benchmarks, too. We could not compare LRT-NG with C2E2, as the latter constructs its reachtubes in relation to a given set of unsafe states. In future work we plant to also compare LRT-NG with CORA~\cite{Althoff2018b}.

All tools were run on a MacBook Pro macOS Mojave $10.14.5$, with a $2.9$ GHz Intel Core i5 processor and $16$GB memory. For each tool we checked the volume of the reachtube, and the time it took to construct it. 
The volumes are given for each tool, benchmark, initial states, time increment, and  integration orders one, two and four, in Table~\ref{table:res}.

LRT-NG with 1st order is clearly tighter than Flow* and CAPD with 2nd order on all benchmarks, except Dubins Car. In the Robotarm, cardiac-cell, cartpole and quadcopter models, we have the best accuracy compared to Flow* and CAPD. Moreover, LRT-NG with RK1 is more accurate than the others using 4th order. Especially on the quadcopter model, and also on the cartpole model, we have by far the tightest reachtube. Flow* has the best accuracy for Brusselator, Van der Pol and Dubins Car, but LRT-NG with RK1 is not far away in terms of accuracy in the first two models. LRT-NG is the only one able to verify the cartpole controlled by a Neural ODE until $t=1$ and the cartpole controlled by an LTC recurrent neural network until $t=0.57$, which shows that we are on a good way of being able to verify cyber-physical system with Neural ODEs as controllers. The improvement over LRT is also clearly shown in Table~\ref{table:res}.


LRT-NG is moderately faster in all benchmarks than Flow*. However, CAPD is much faster compared to both. CAPD was developed over the past decade, and the current version is incredibly optimized, especially with respect to interval arithmetic. We use instead the public IBEX library~\cite{ibex}. We plan to further optimizate LRT-NG in the future.

For the details of all of the used benchmarks, i.e. the system of ODEs and interpretation of the variables, we defer to the appendix.

\begin{table}[t]
\scriptsize

\centering
\hspace*{-2ex}
\caption{\small Performance comparison with Flow*, CAPD and LRT. We use following labels for models M: B(2)- Brusselator, V(2)- Van der Pol oscillator, R(4)- Robotarm, D(3)- Dubins Car, M(2)- Mitchell Schaeffer cardiac-cell, C(4)- controlled cartpole, Q(17)-Quadcopter, C-N(12)- cartpole wt. Neural ODE, C-L(12)- cartpole wt. LTC RNN (number in parenthese denotes dimension). T: time horizon, dt: time step, r: initial radius in each dimension, AV: average volume of reachtubes in T, Fail: Volume blow-up before T; (1), (2) and (4) denotes the integration order. We mark in bold the best performers for low orders (1,2nd) and the higher order (4th).}
\tiny
\begin{center}
\bgroup
\def\arraystretch{2}%
\begin{tabular}{|C{0.3cm}|C{0.6cm}|C{0.25cm}|C{0.4cm}| C{1.2cm} |C{1.1cm}|C{1cm}|C{1.1cm}|}
\hline
\multirow{2}{*}{M}
& \multirow{2}{*}{dt}
& \multirow{2}{*}{T}
& \multirow{2}{*}{r}
&	
    \multicolumn{4}{c|}{AV}
      \\
      \cline{5-8}
     &&&& LRT-NG &  Flow* & CAPD & LRT\\
	\hline 

\multirow{2}{*}{B(2)}
& \multirow{2}{*}{$0.01$}
& \multirow{2}{*}{$9$}
& \multirow{2}{*}{$0.01$}
& $\mathbf{1.5\text{e-}{4}}$ \textbf{(1)} & $5.1\text{e-}{3}$ (2)& $4.3\text{e-}{4}$ (2) & $6.7\text{e-}{4}$ (1)\\
&&&& $1.4\text{e-}{4}$ (2, 4) & $\mathbf{9.8\textbf{e-}{5}}$ \textbf{(4)} & $3.6\text{e-}{4}$ (4) &$6.1\text{e-}{4}$ (2, 4)
\\
\hline

\multirow{2}{*}{V(2)}
& \multirow{2}{*}{$0.01$}
& \multirow{2}{*}{$40$}
& \multirow{2}{*}{$0.01$}
& $\mathbf{4.2\text{e-}{4}}$ \textbf{(1)} & $5.6\text{e-}{3}$ (2) & $1.5\text{e-}{3}$ & $4.1\text{e-}{3}$ (1)\\
&&&& $4.1\text{e-}{4}$ (2, 4) & $\mathbf{3.5\textbf{e-}{4}}$ \textbf{(4)} & (2, 4) & $3.7\text{e-}{3}$ (2, 4)
\\
\hline

\multirow{2}{*}{R(4)}
& \multirow{2}{*}{$0.01$}
& \multirow{2}{*}{$40$}
& \multirow{2}{*}{$0.005$}
& $\mathbf{8\text{e-}{11}}$ \textbf{(1,2)} & $1.1\text{e-}{9}$ (2) & $1.1\text{e-}{9}$ & \multirow{2}{*}{Fail}\\
&&&& $\mathbf{7.9\textbf{e-}{11}}$ \textbf{(4)} & $8.7\text{e-}{10}$ (4) & (2, 4) &
\\
\hline

\multirow{2}{*}{D(3)}
& \multirow{2}{*}{$0.00125$}
& \multirow{2}{*}{$15$}
& \multirow{2}{*}{$0.01$}
& $0.132$ (1) & $6.6037$ (2) & $\mathbf{0.1181}$ & $390$ (1)\\
&&&& $0.131$ (2, 4) & $\mathbf{4.5\textbf{e-}{2}}$ \textbf{(4)} & \textbf{(2, 4)} & $385$ (2, 4)
\\
\hline

\multirow{2}{*}{M(2)}
& \multirow{2}{*}{$0.01$}
& \multirow{2}{*}{$10$}
& \multirow{2}{*}{$10^{-4}$}
& $\mathbf{3.8\text{e-}{9}}$ \textbf{(1, 2)} & $3.9\text{e-}{8}$ (2) & $4.9\text{e-}{8}$ (2) & $3.2\text{e-}{8}$\\
&&&& $\mathbf{3.7\textbf{e-}{9}}$ \textbf{(4)} & $1.5\text{e-}{8}$ (4) & $4.4\text{e-}{8}$ (4) & (1, 2, 4)
\\
\hline

\multirow{2}{*}{C(4)}
& \multirow{2}{*}{$0.001$}
& \multirow{2}{*}{$10$}
& \multirow{2}{*}{$10^{-4}$}
& $\mathbf{8.4\text{e-}{17}}$ \textbf{(1)} & $1.1\text{e-}{11}$ (2) & $2.6\text{e-}{13}$ (2) & \multirow{2}{*}{Fail}\\
&&&& $\mathbf{7.2\textbf{e-}{17}}$ \textbf{(2,4)} & $7\text{e-}{13}$(4) & $2.6\text{e-}{13}$ (4) &
\\
\hline

\multirow{2}{*}{Q(17)}
& \multirow{2}{*}{$10^{-4}$}
& \multirow{2}{*}{$2$}
& \multirow{2}{*}{$0.005$}
& $\mathbf{3.21\text{e-}{54}}$ \textbf{(1)} & \multirow{2}{*}{$9.7\text{e-}{25}$ (2)} & \multirow{2}{*}{$1.7\text{e-}{31}$ (2)} & \multirow{2}{*}{Fail}\\
&&&& $\mathbf{9.31\textbf{e-}{56}}$ \textbf{(2)} & & &
\\
\hline

C-N(12) & $10^{-5}$ & $1$ & $10^{-4}$
& $\mathbf{3.9\text{e-}{27}}$ \textbf{(1, 2)}
& Fail
& Fail
& Fail
\\
\hline

C-L(12) & $10^{-6}$ & $0.35$ & $10^{-4}$
& $\mathbf{4.49\text{e-}{33}}$ \textbf{(1)}
& Fail
& Fail
& Fail
\\
\hline

\end{tabular}
\egroup
\vspace{1ex}
\vspace*{1ex}
\end{center}
\label{table:res}
\vspace{-10ex}
\end{table}

\textbf{Brusselator.}
Fig.~\ref{fig:Brusselator} compares the reachtubes computed by Flow* (in green), CAPD (in orange), and LRT-NG (in violet). Fig.~\ref{fig:BrusselatorEntire} illustrates the reachtubes in their entirety while the Fig.~\ref{fig:BrusselatorZoomIn} zooms in at the last portion of the reachtubes. Flow* blows up already before time $t=10$, while CAPD works up to $t=19.1$ and LRT-NG up to $t=22.6$. 
\begin{figure}
    \centering
    \begin{subfigure}[b]{0.45\textwidth}
        \centering
        \includegraphics[width=0.8\textwidth]{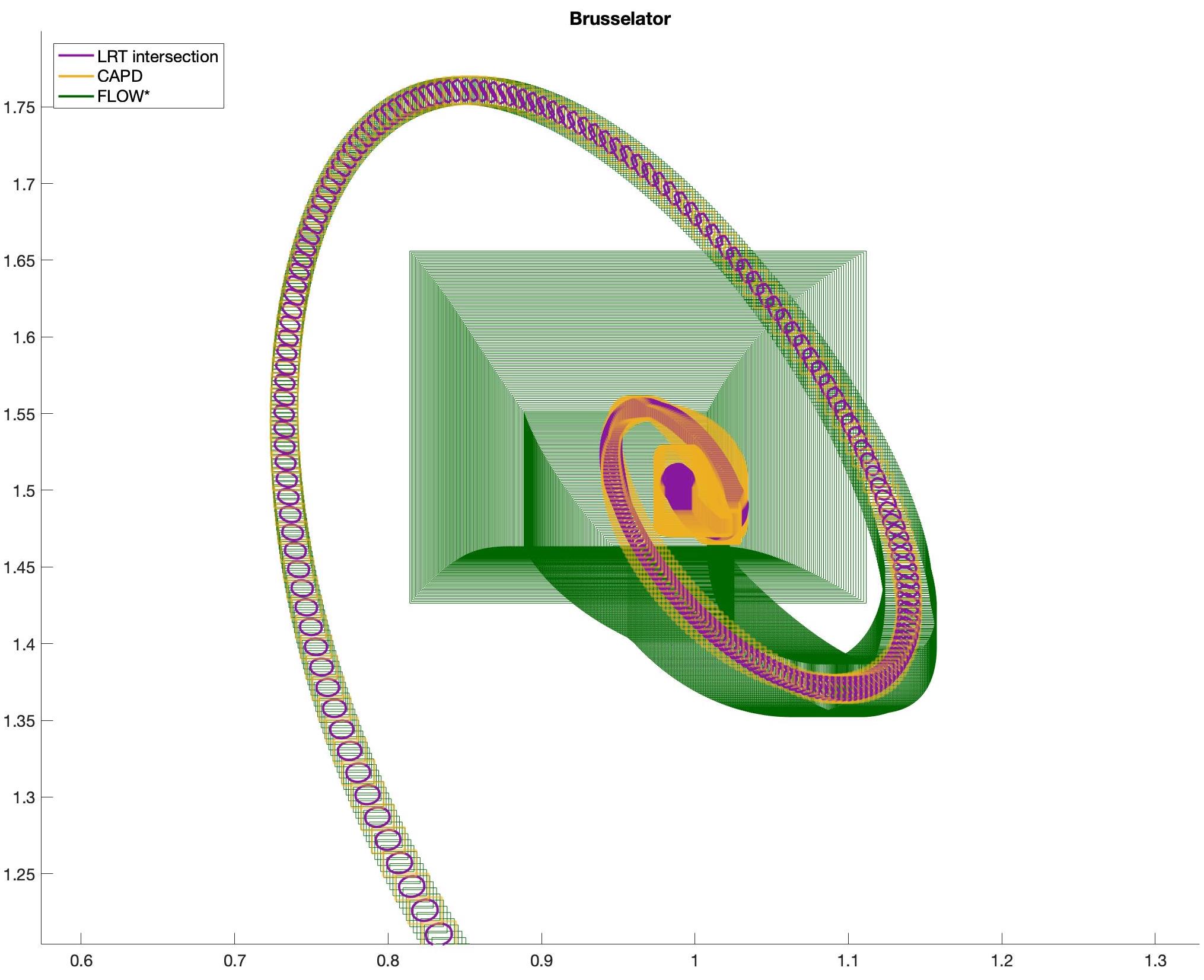}
        \caption{Entire Reachtube.}
        \label{fig:BrusselatorEntire}
    \end{subfigure}
    \begin{subfigure}[b]{0.45\textwidth}
        \centering
        \includegraphics[width=0.8\textwidth]{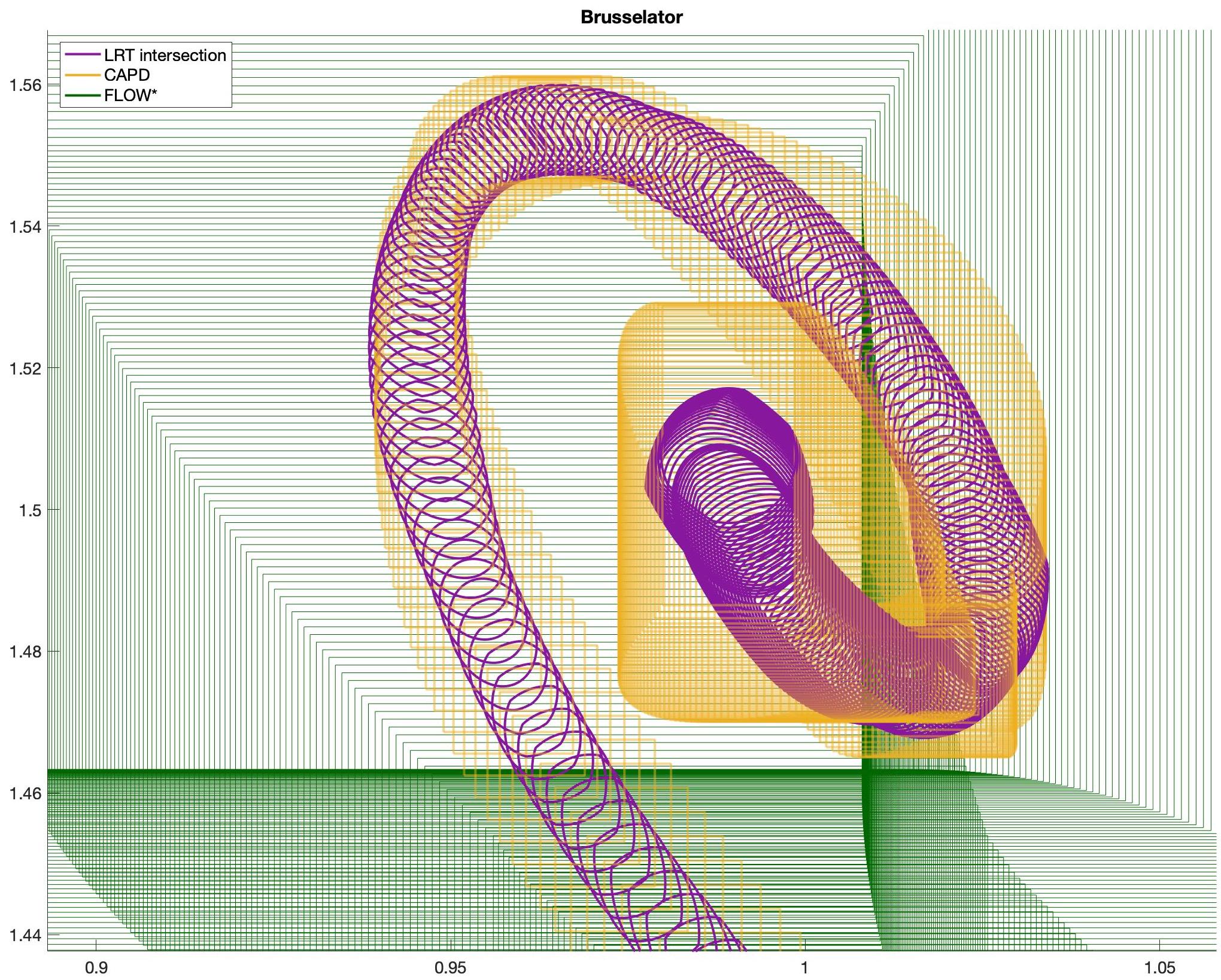}
        \caption{Reachtube zoomed to last portion.}
        \label{fig:BrusselatorZoomIn}
    \end{subfigure}
    \caption{Brusselator model until time $t=14$.}
    \label{fig:Brusselator}
\end{figure}
As clearly shown in these figures, LRT-NG is superior to both Flow* and CAPD on this Brusselator model.

\textbf{Van der Pol Oscillator.}
The comparison of the reachtubes constructed by Flow*, CAPD, and LRT-NG is shown in Fig.~\ref{fig:VanDerPol}.
\begin{figure}
    \centering
    \includegraphics[width=0.34\textwidth]{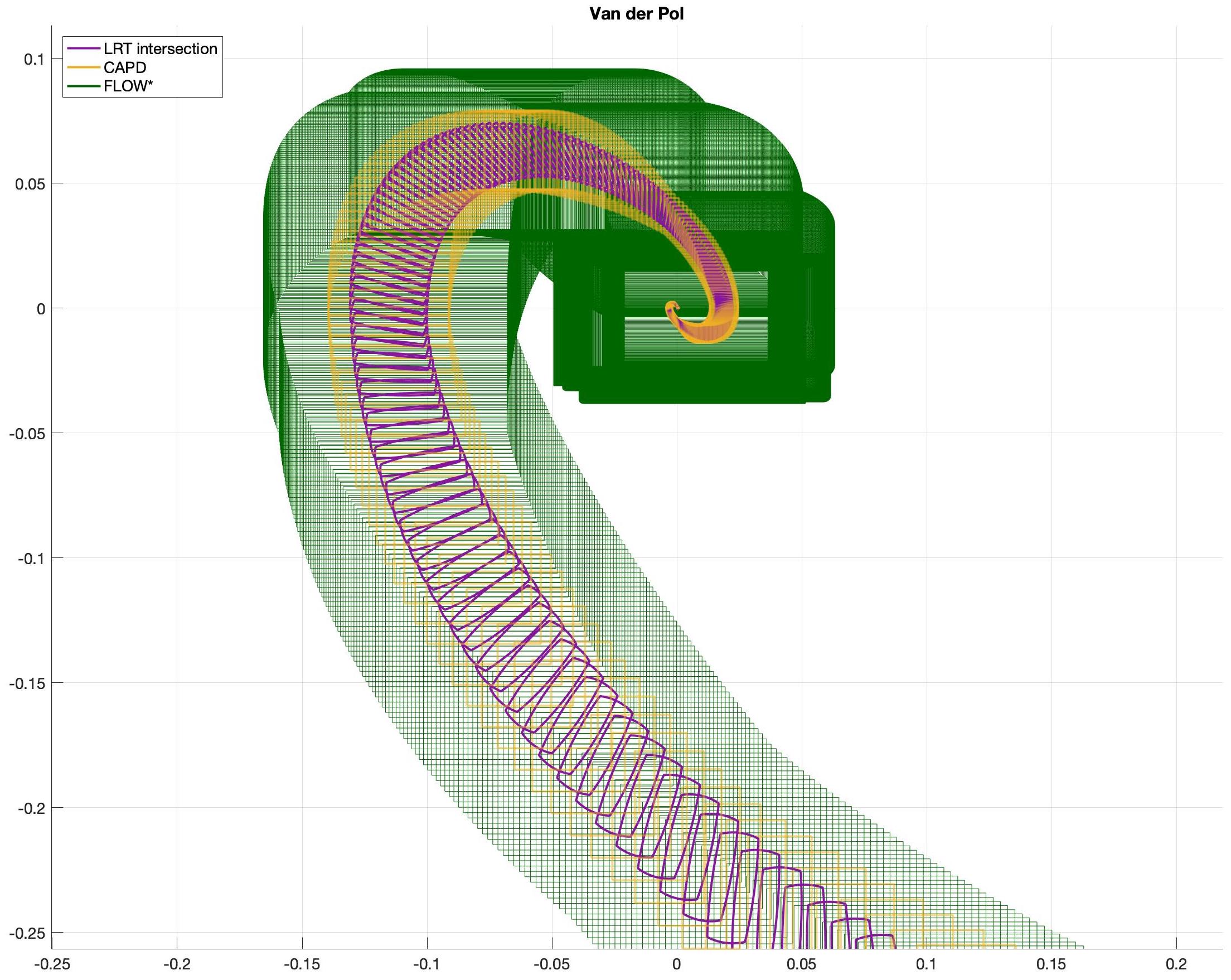}
    \caption{Van der Pol model until time $t=40$.}
    \label{fig:VanDerPol}
    \vspace*{-3ex}
\end{figure}
\begin{figure}
     \centering
     \includegraphics[width=0.34\textwidth]{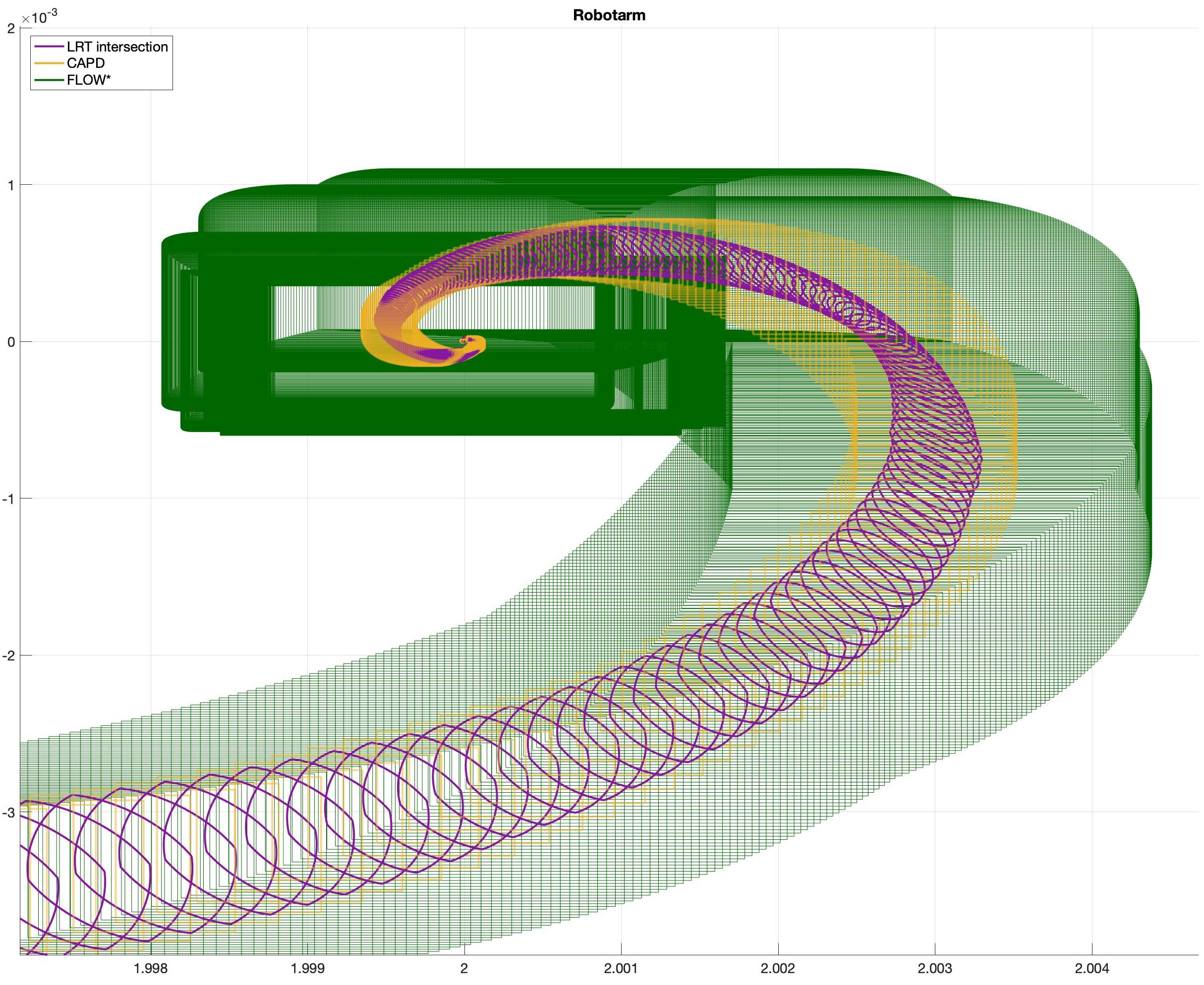}
     \caption{Robotarm model until time $t=40$.}
     \label{fig:Robotarm}
    \vspace*{-3ex}
\end{figure}
The figure zooms in on the last time-segment of the reachtube, where it is clearly visible that LRT-NG is superior to both Flow* and CAPD on this benchmark.

\textbf{Robotarm.}
The comparison of the reachtubes constructed by Flow*, CAPD, and LRT-NG is shown in Fig.~\ref{fig:Robotarm}.
The figure shows the projection of the last segment onto the $(x_1,x_4)$-plane. As one can see, LRT-NG is superior both to Flow* and to CAPD, and CAPD seems to be superior to Flow*. Although CAPD and LRT-NG have similar volumes for some time-steps, in other region (e.g. the one on the upper right side of the picture), LRT-NG is much tighter and thus more stable during time, which is an important property to avoid false positives when intersecting with bad states.

\textbf{Dubins Car.}
The comparison of the reachtubes constructed by Flow*, CAPD, and LRT-NG is illustrated in Fig.~\ref{fig:DubinsCar}. The figure illustrates the projection of the reachtubes onto the $(x_1\,{,}\,x_2)$-plane.
As one can see, Flow* blows up and CAPD seems to beat LRT-NG.
\begin{figure}[t]
    \centering
    \includegraphics[width=0.34\textwidth]{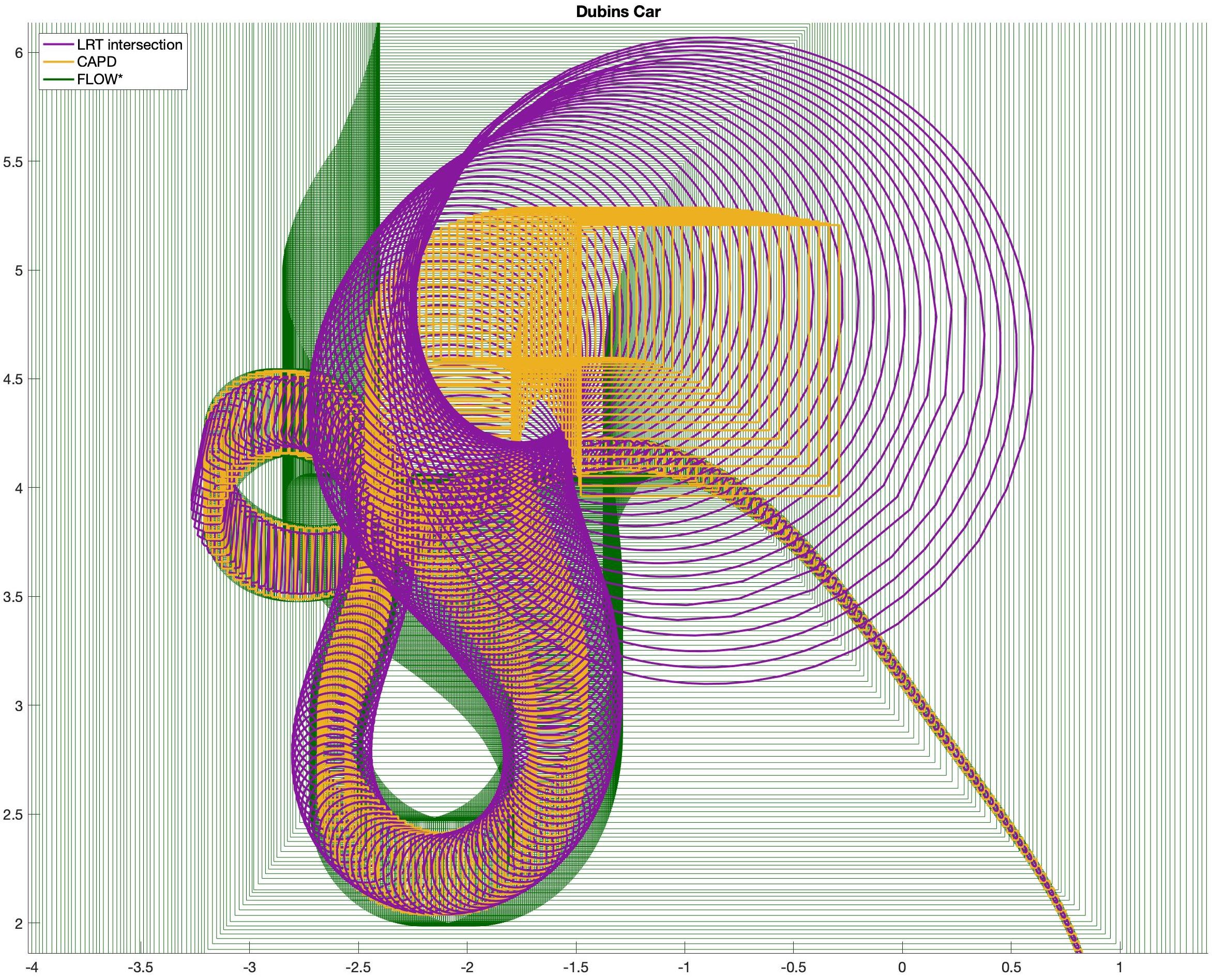}
    \caption{Dubins-Car model until time $t=15$.}
    \label{fig:DubinsCar}
\end{figure}
Hence, we consider it fair claim that CAPD performs consistently better than Flow* and LRT-NG on this benchmark.

\textbf{Cardiac-cell model.}
\begin{figure}
    \centering
    \begin{subfigure}[b]{0.45\textwidth}
        \centering
        \includegraphics[width=0.7\textwidth]{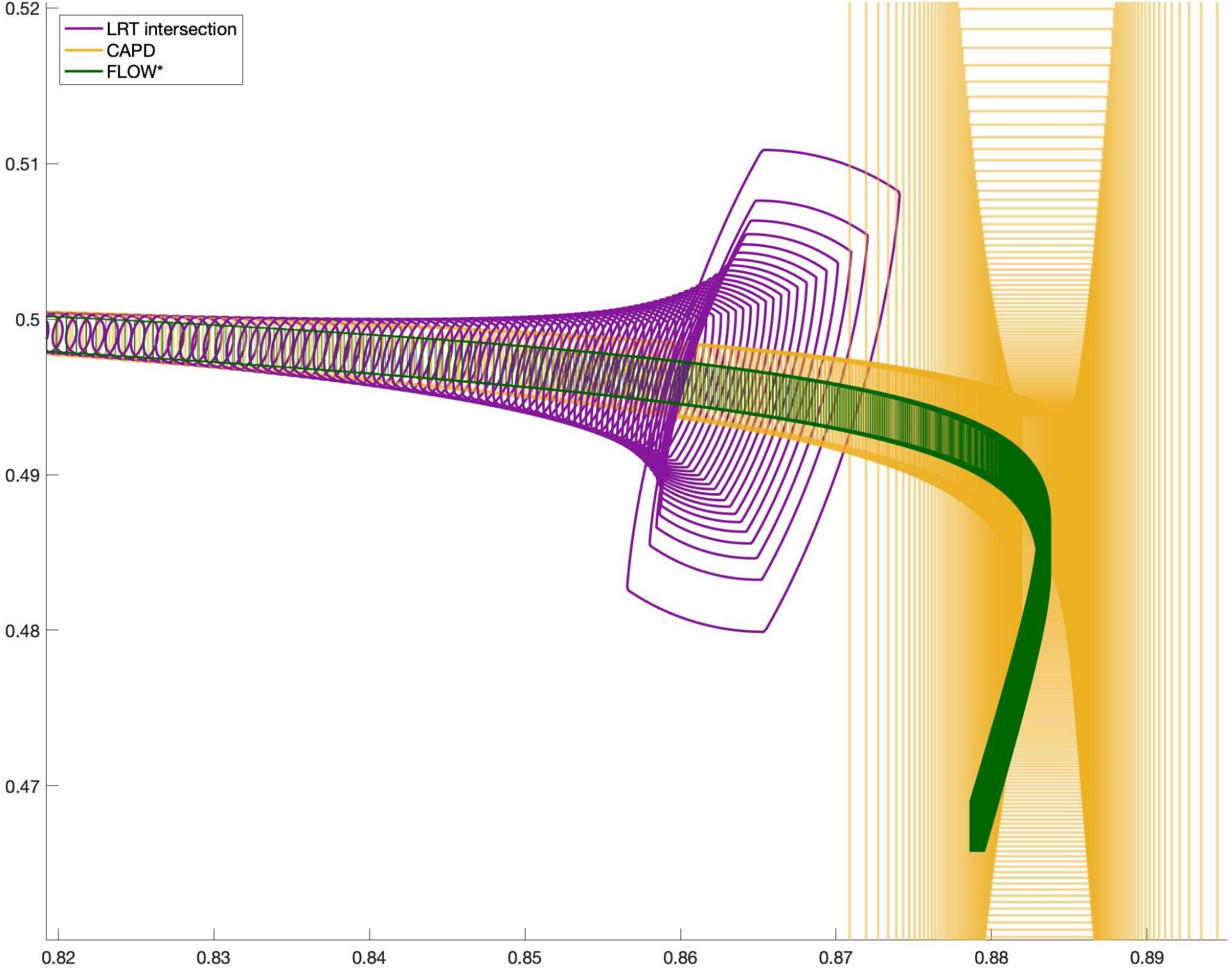}
        \caption{initial radius $\rd_0=0.001$}
        \label{fig:MSa}
    \end{subfigure}
    \begin{subfigure}[b]{0.45\textwidth}
        \centering
        \includegraphics[width=0.8\textwidth]{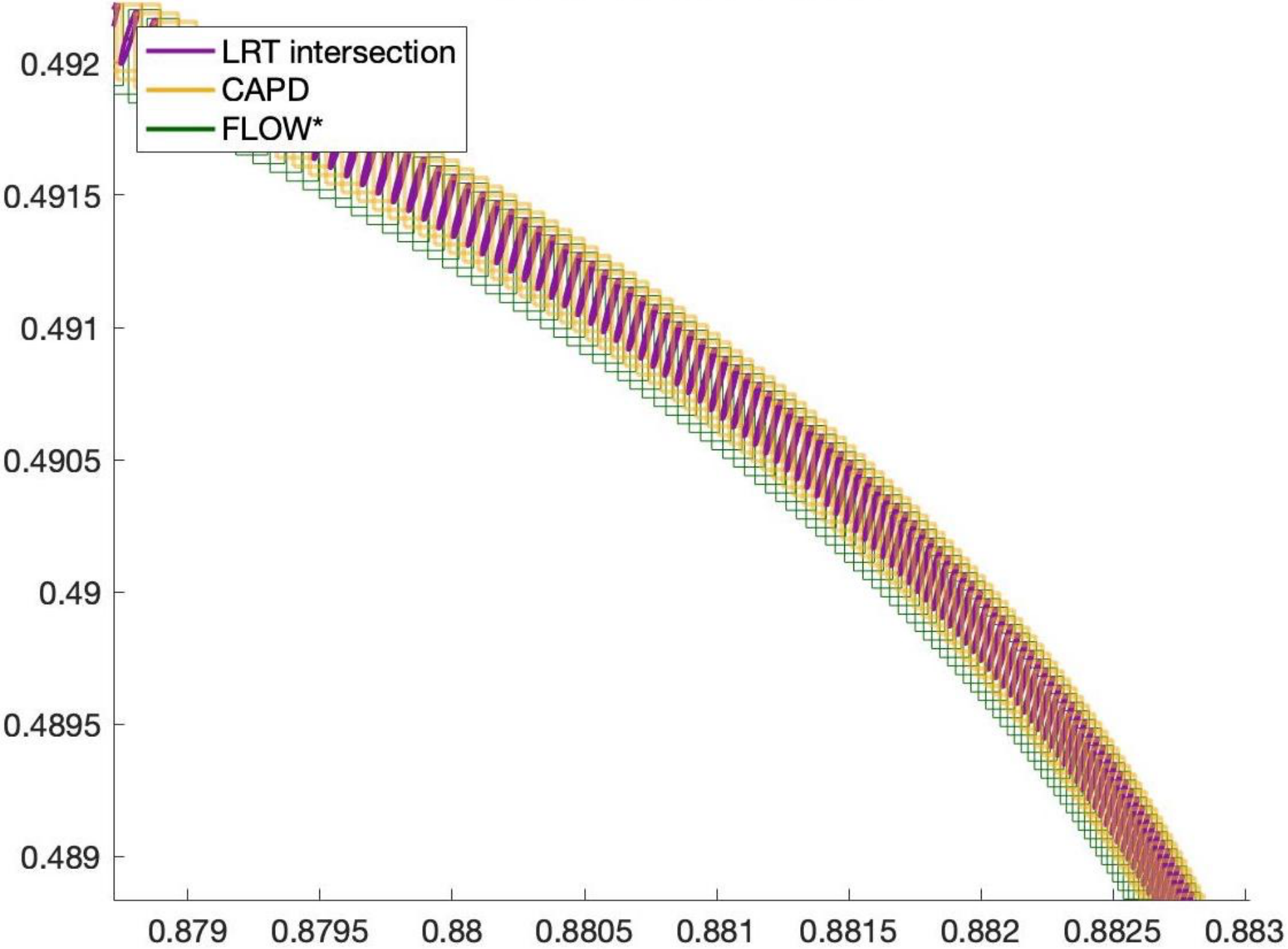}
        \caption{initial radius $\rd_0=10^{-4}$}
        \label{fig:MSb}
    \end{subfigure}
    \caption{Mitchell Schaeffer cardiac-cell model until $t\,{=}\,10$ with different initial radii.}
    \label{fig:MS}
    \vspace*{-3ex}
\end{figure}
The comparison of the reachtubes is shown in Fig.~\ref{fig:MS}. In Fig.~\ref{fig:MSa} we illustrate the reachtube for the initial radius $0.001$, where LRT-NG is the first one two blow up, followed by CAPD, in contrast to Flow* which stays tight until time $t=10$. Fig.~\ref{fig:MSb} illustrates the same model but with the smaller initial radius $10^{-4}$ and zoomed in on the last segment, where LRT-NG produces the tightest tube, followed by CAPD and Flow*.

\textbf{Cartpole with linearly-stabilizing controller.}

\begin{figure}
    \centering
    \includegraphics[width=0.34\textwidth]{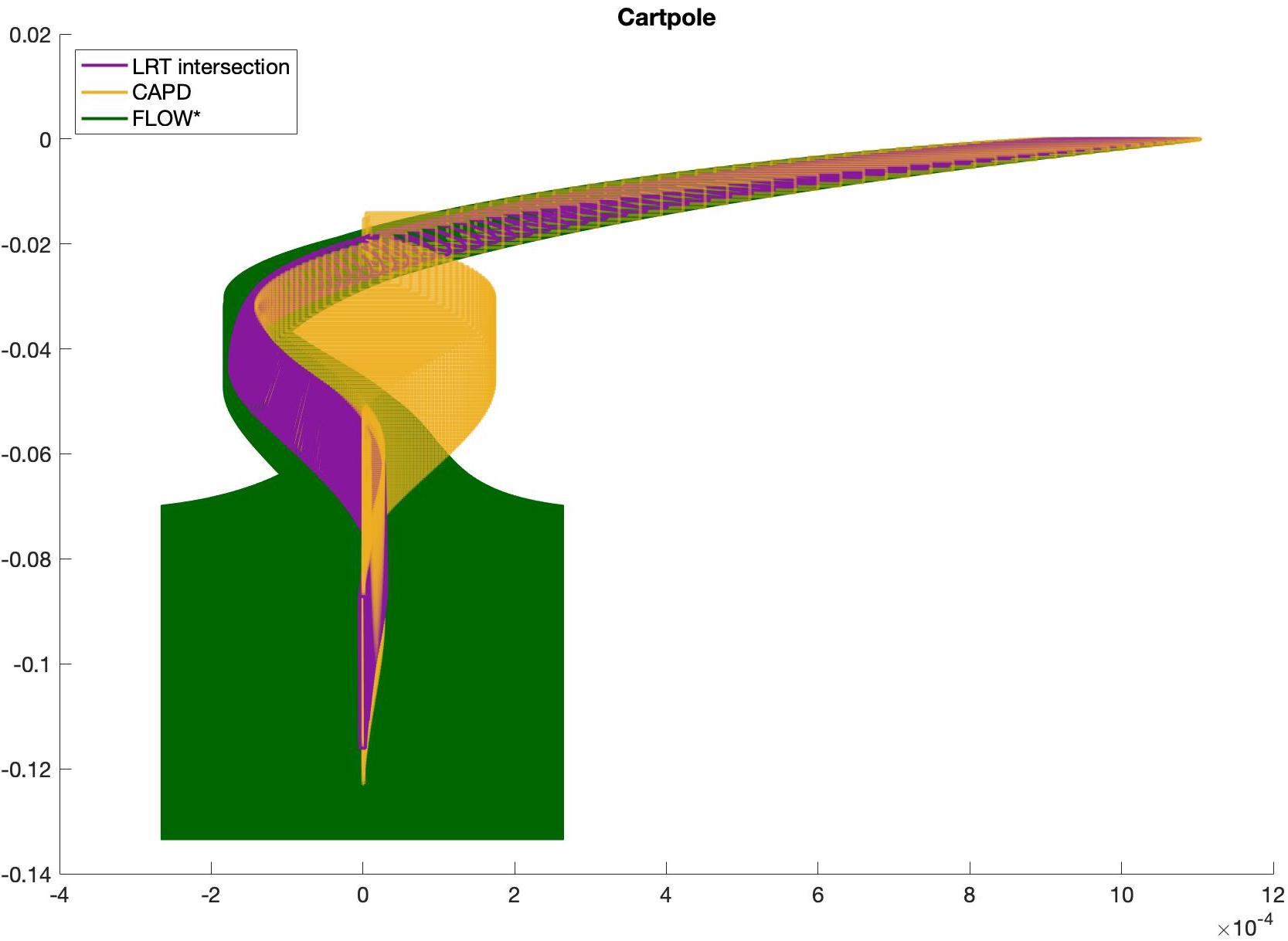}
    \caption{Cartpole model until time $t=10$.\label{fig:Cartpole}}
\end{figure}
\begin{figure}
    \centering
    \includegraphics[width=0.34\textwidth]{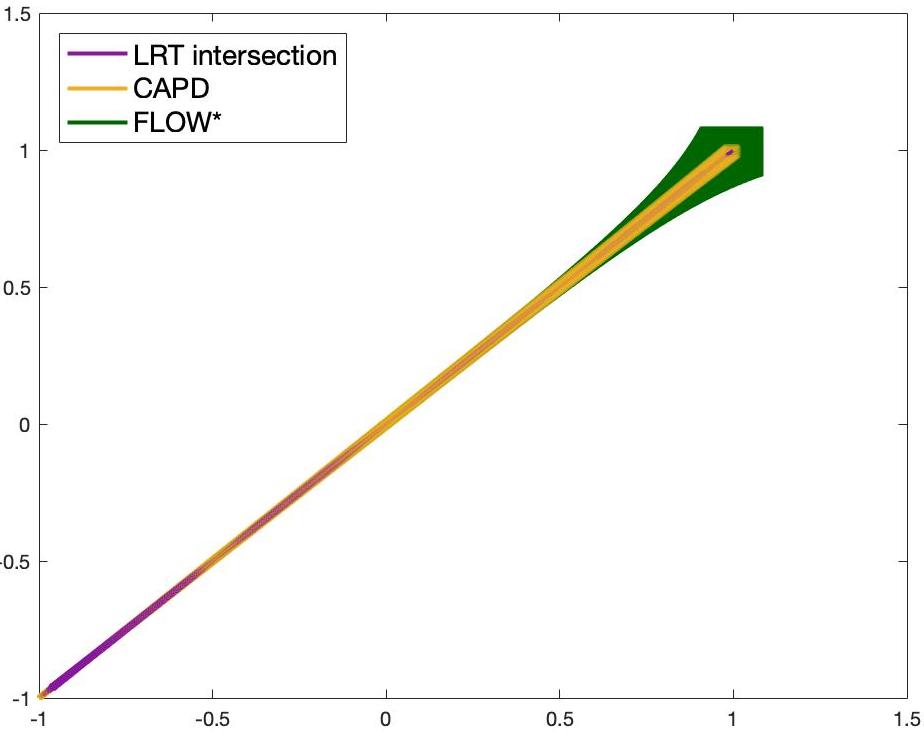}
    \caption{Quadcopter model until time $t=2$.\label{fig:Quadcopter}}
    \vspace*{-3ex}
\end{figure}
The comparison of the tools until time $t=10$ is shown in Fig.~\ref{fig:Cartpole}. The picture illustrates a projection of the reachtubes onto the $(\theta,x)$-plane, thus the pictures shows the angle of the pole being stabilized by moving the cart position. As one can see, LRT-NG has the tightest reachtubes, followed by CAPD and Flow*.

\textbf{Quadcopter model.}
Fig.~\ref{fig:Quadcopter} illustrates the projection of the reachtubes onto the $(p_n,p_e)$-plane, showing that LRT-NG has by far the tightest reachtubes compared to CAPD and Flow*.

\textbf{Cartpole with Neural ODEs as controller.}
We developed and trained two different continuous-time recurrent neural networks to control a cartpole: a Neural ODE~\cite{neuralODEs} and an LTC recurrent neural network~\cite{icra19}.
\begin{figure}
\vspace*{2ex}
    \centering
    \begin{subfigure}[b]{0.45\textwidth}
        \centering
        \includegraphics[width=0.9\textwidth]{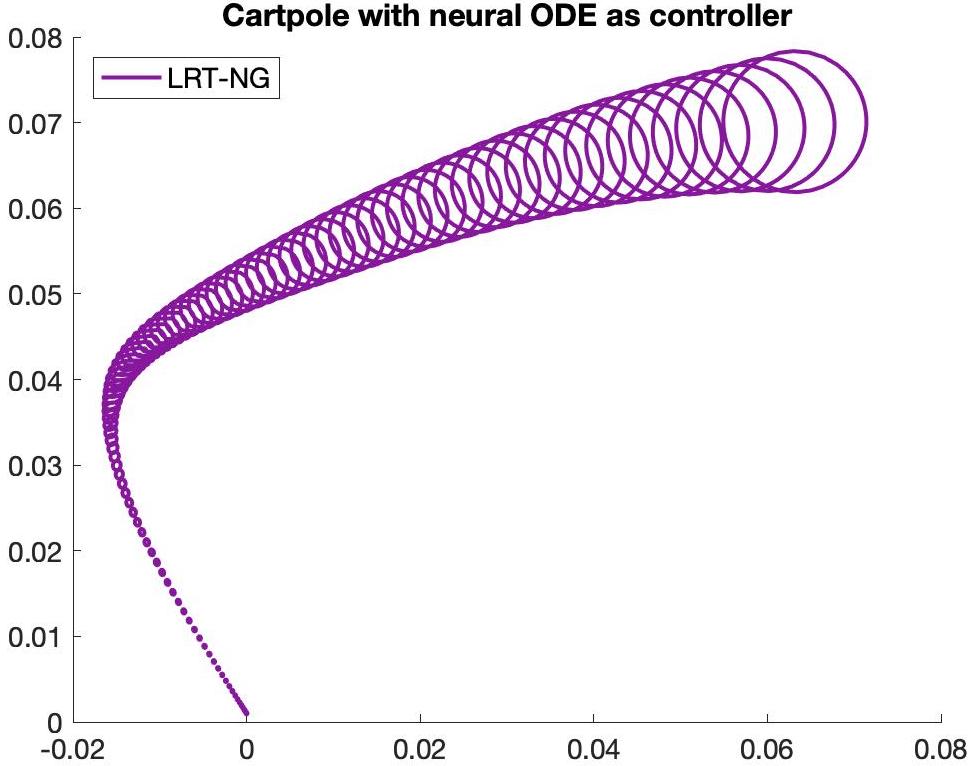}
        \caption{Neural ODE controller until $t\,{=}\,1$}
        \label{fig:cartpoleNeuralODE}
        \vspace*{1ex}
    \end{subfigure}
    \begin{subfigure}[b]{0.45\textwidth}
        \centering
        \includegraphics[width=0.9\textwidth]{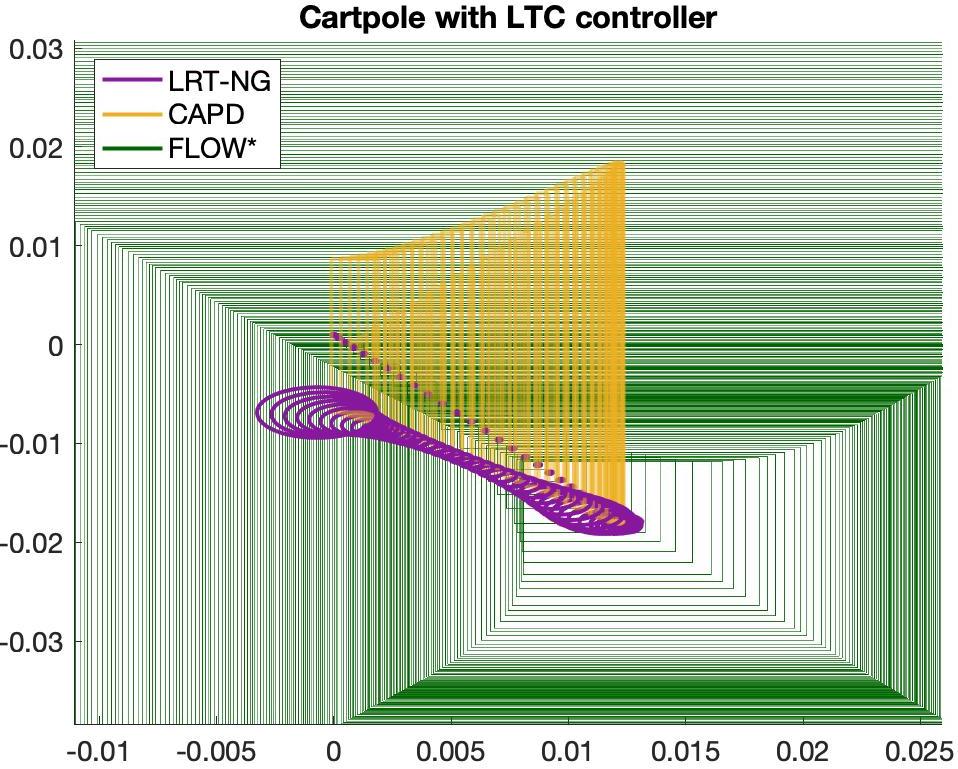}
        \caption{LTC recurrent neural network controller until $t\,{=}\,0.35$}
        \label{fig:cartpoleLTC}
    \end{subfigure}
    \caption{Cartpole with two different Neural ODEs as controller}
    \label{fig:cartpoleDNN}
    \vspace*{-3ex}
\end{figure}
Neural ODE controller: as Flow* and CAPD are not able to construct a reachtube for a longer time horizon than $t=0.135$, we show in Fig.~\ref{fig:cartpoleNeuralODE} only the reachtube constructed by LRT-NG onto the $(x,\theta)$-plane until $t=1$.
LTC (Neural ODE) controller: Fig.~\ref{fig:cartpoleLTC} illustrates the reachtube for the same cartpole system dynamics controlled by an LTC neural network. Our tool was able to run until $t=0.57$, Flow* blows up before $t=0.29$ with an average volume of $2.3\text{e-}1$ and CAPD blows up before $t=0.33$ with an average volume of $2.6\text{e-}18$.

\section{Conclusions and Future Work}\label{conclusions}

We introduceded LRT-NG, a significant improvement of LRT in terms of both
theoretical underpinning and tool implementation. As our experiments show, LRT-NG is superior to LRT, CAPD, and Flow*. Unlike previous bloating-based reachability analysis tools, LRT-NG uses an explicit analytic solution for the optimal metric, minimizing the ellipsoidal reachset. Also, the reachset computation is based on intersections between balls and ellipsoids, which considerably reduces the wrapping effect. Finally, LRT-NG avoids the wrapping effect occurring in the validated integration of the center of the reachset by optimally absorbing the interval approximation in the radius of the ball. Conservative methods still have a problem with high-dimensional stiff dynamics for long time horizons, this is why an interesting direction for future work is to explore complementary stochastic approaches. The ultimate goal is to scale up to large NeuralODE controllers for complicated CPS or biological tasks. 





\section*{Acknowledgment}
The authors would like to thank Ramin Hasani and Guillaume Berger for intellectual discussions about the research which lead to the generation of new ideas.
ML was supported in part by the Austrian Science Fund (FWF) under grant Z211-N23 (Wittgenstein Award).  Smolka's research was supported by NSF grants CPS-1446832 and CCF-1918225. Gruenbacher is funded by FWF project W1255-N23. JC was partially supported by NAWA Polish Returns grant PPN/PPO/2018/1/00029.




\bibliography{root}
\bibliographystyle{abbrv}

\newpage
\appendices

\section{Details of the used benchmarks}
\addcontentsline{toc}{section}{Details of the used benchmarks}
\label{secbenchmarks}
In this appendix we present the details of all of the  benchmarks used in the paper, i.e. the system of ODEs and the interpretation of the variables.

\subsection{Brusselator}
The Brusselator is a two ODEs theoretical model for auto-catalytic reactions, where $x_1$ and $x_2$ are chemical concentrations:
{\small
\begin{align*}
    \partial_tx_1 &= 1 + x_1^2x_2 - 2.5x_1 &x_1(t_0) = 1\\
    \partial_tx_2 &= 1.5x_1-x_1^2x_2 &x_2(t_0) = 1
\end{align*}
}

\subsection{Van der Pol Oscillator}
The two ODEs below describe a non-conservative oscillator with non-linear damping, whose position is $x_1$, and velocity is $x_2$:
{\small
\begin{align*}
    \partial_tx_1 &= x_2 &x_1(t_0) = -1\\
    \partial_tx_2 &= (x_1^2-1)x_2-x_1 &x_2(t_0) = -1
\end{align*}
}

\subsection{Robotarm}
The Robotarm~\cite{robot_arm_2000} is described by the four ODEs below, where $x_1$ is the angle, $x_2$ is the position, $x_3$ is the angular velocity and $x_4$ the velocity of the robotarm:
{\small
\begin{align*}
    \partial_t x_1 &= x_3 &x_1(t_0) = 1.505\\
    \partial_t x_2 &= x_4 &x_2(t_0) = 1.505\\
    \partial_t x_3 &= \frac{-2 x_2 x_3 x_4 - 2 x_1 - 2 x_3}{x_2^2 + 1} + \frac4{x_2^2 + 1} &x_3(t_0) = 0.005\\
    \partial_t x_4 &= x_2 x_3^2 - x_2 - x_4 + 1 &x_4(t_0) = 0.005
\end{align*}
}

\subsection{Dubins Car}
The Dubins Car can either go forward, turn fully left, or turn fully right. Its position is given by $x_1$ and $x_2$, and its pose by $x_3$. As the pose is time variant, we added $t$ as an additional variable. LRT-NG does not include the time variable in the reachtube construction, and it allows the user to explicitly point out the time variable:

{\small
\begin{align*}
    \partial_t x_1 &= \cos(x_3) &x_1(t_0) = 0\\
    \partial_t x_2 &= \sin(x_3) &x_2(t_0) = 0\\
    \partial_t x_3 &= x_1 \sin(t) &x_3(t_0) = 0.7854\\
    \partial_t t &= 1 &t(t_0) = 0
\end{align*}}

\subsection{Cardiac-cell model}
The Mitchell-Schaeffer model~\cite{mitchell2003two} is a simplified description of the electrical activity of the cardiac-cells membrane. The model incorporates only two state variables: the electric potential of the cell membrane $x_1$, and the gate dynamics of the sodium ion channels $x_2$. The continuous version of the model is described by the following two nonlinear ODEs: 
{\small 
\begin{align*}
    \partial_tx_1 &= x_2 x_1^2 \frac{1 - x_1}{0.3} - \frac{x_1}{6} &x_1(t_0) = 0.8\\
    \partial_tx_2 &= \sigma(x_1)\frac{-x_2}{150} + (1-\sigma(x_1)) \frac{1 - x_2}{20} &x_2(t_0) = 0.5\\
    \textrm{with } \sigma(x_1) &= 0.5 (1+\tanh{ (50x_1-5) })
\end{align*}}

\subsection{Cartpole with linear stabilizing controller}
The cartpole model coupled with a stabilizing linear controller is described by the following ODEs.
{\small
\begin{align*}
    \partial_t\sigma &= \frac{f \cos(\theta) - ml\sigma^2\cos(\theta)\sin(\theta) + (m+M)g\sin(\theta)} {l(M+m\sin^2(\theta)},\\
    \partial_tw &= \frac{f + m\sin(\theta) * (-l\sigma^2 + g\cos(\theta))}{M+m\sin^2(\theta)},\\
    \partial_tx &= w,\\
    \partial_t\theta &= \sigma,\\
    f &= -1.1Mg\theta - \sigma.
\end{align*}}
We use the standard cartpole model, with an external controller ($f$). The variables have the following interpretation: $x$ is the position of the cart, $\theta$ is the angle of the pole with respect to the cart, $w$ is the velocity of the cart, and $\sigma$ is the velocity of the pole. The constants $m, M, l, g$ denote the mass of the pole and the cart, the length of the pole and the gravitation constant respectively. The forcing $f$ is selected such that it is stabilizing the pole angle (the angle converges to $0$ from some initial angle $\theta_0$). As the initial condition we use $x_0=0, w_0=0, \theta_0=0.001, \sigma_0=0$ (the center of the ball), and set the constants $M=1, m=0.001, g=9.81, l=1$.

\subsection{Quadcopter model}
The quadcopter model is described by the following ODEs
{\small
\begin{align*}
\partial_tp_n &= 2u(q_0^2+q_1^2-0.5) - 2v(q_0q_3-q_1q_2) + 2w(q_0q_2+q_1q_3)\\
\partial_tp_e &= 2v(q_0^2+q_2^2-0.5) + 2u(q_0q_3+q_1q_2)-2w(q_0q_1-q_2q_3)\\
\partial_th &= 2w(q_0^2+q_3^2-0.5) - 2u(q_0q_2-q_1q_3)+2v(q_0q_1+q_2q_3)\\
\partial_tu &= rv-qw-11.62(q_0q_2-q_1q_3)\\
\partial_tv &= pw-ru+11.62(q_0q_1+q_2q_3)\\
\partial_tw &= qu-pv+11.62(q_0^2+q_3^2-0.5)\\
\partial_tq_0 &= -0.5q_1p-0.5q_2q-0.5q_3r\\
\partial_tq_1 &= 0.5q_0p - 0.5q_3q+0.5q_2r\\
\partial_tq_2 &= 0.5q_3p + 0.5q_0q-0.5q_1r\\
\partial_tq_3 &= 0.5q_1q - 0.5q_2p + 0.5 q_0r\\
\partial_tp &= (-40.0006326p_I-2.82839798295p)-1.1334074237qr\\
\partial_tq &= (-39.9998045q_I-2.82837525410q)+1.1320781796pr\\
\partial_tr &= (-39.9997891r_I-2.82841342233r)-0.00469522pq\\
\partial_tp_I &= p,\  \partial_tq_I = q,\  \partial_tr_I = r,\  \partial_th_I = h
\end{align*}}
\noindent We do not include an external controller in the model. The variables have the following interpretations; $(p_n,p_e,h)$ is the inertial position, $(u,v,w)$ is the linear velocity, $(q_0,q_1,q_2,q_3)$ is the angular orientation quaternion, $(p,q,r)$ is the angular velocity, and $(p_I, q_I, r_I)$ are the integral states. As the set of initial states we set
$p_n\in[-1,-0.99], p_e\in[-1,-0.99], h\in[9, 9.01], u\in[-1,-0.99], v\in[-1,-0.99], w\in[-1,-0.99], q_0=0, q_1=0, q_2=0, q_3=1, p\in[-1,-0.99], q\in[-1,-0.99], r\in[-1,-0.99], p_I=0, q_I=0, r_I=0, h_I=0$. For details refer \cite{quad}.

\subsection{Cartpole with Neural ODE controller}
We used the gym environment to train the neural network controller and thus the cartpole's dynamic is modelled by the following nonlinear ODEs~\cite{Barto83} (considering that there is neither a friction of the cart on the track nor of the pole on the cart):
{\small
\begin{align*}
    \partial_t\sigma &= \frac{g \sin(\theta) + \cos(\theta)
    \left(
    \frac{-F-ml(\partial_t\theta)^2\sin(\theta)}{m_c + m}
    \right)}
    {l
    \left(
    \frac4{3}-\frac{m\cos^2(\theta)}{m_c + m}
    \right)
    } &\sigma(t_0) = 0\\
    \partial_tw &= \frac{F + ml
    \left(
    (\partial_t\theta)^2\sin(\theta)-(\partial_t\sigma)\cos(\theta)
    \right)}{m_c + m} &w(t_0) = 0\\
    \partial_tx &= w &x(t_0) = 0\\
    \partial_t\theta &= \sigma &\theta(t_0) = 0.001
\end{align*}}
with the constants set to $g=9.81$ (acceleration due to gravity), $m_c = 1$ (mass of cart), $m=0.1$ (mass of pole), $l=0.5$ (half-pole length) and $F\in[-10,10]$ (force applied to cart's center of mass), which is determined by the chosen action of the Neural ODE controller.

The \emph{Neural ODE controller} we trained has one layer with eight neurons $h_i, i\in\{1,\dots,8\}$. Let us define the current state of the environment as $y = (x,w,\theta,\sigma)$ and the current state of the controller as $h = (h_i)_{i\in\{1,\dots,8\}}$. Each neuron is defined by an ODE as follows:
{\small
\begin{align*}
    \partial_th_i = -h_i + \tanh(Wy+Vh+b)
\end{align*}
}
and the final policy is then obtained by
{\small
\begin{align*}
    F = 10 \cdot \tanh(Uh+c).
\end{align*}
}
\noindent{}with weighting matrices $W,V,U$ and biases $b$ and $c$, which are learned during the training process.

The \emph{LTC recurrent neural network} we trained has also one layer with eight neurons as described by the ODE
\label{ltceq}
\begin{align*}
&\partial_t v_i =\frac{1}{c_{i}}\Big( g_{leak,i} \big(v_{leak,i} - v_i \big) + \sum_{j=1}^{8} \frac{w_{ij} (E_{ij} - v_{i})}{(1+ e^{-\sigma_{ij}(v_{j} + \mu_{ij})})}\Big),\\
&v_i(t_0)=0
\end{align*}
for $i \in \{1,\dots 8\}$.
The output of the network is computed by
\begin{align*}
    F = 10 \cdot \tanh\Big(\sum_{j=1}^{8}a_i v_i + b\Big).
\end{align*}
The values of the parameters $c_{i}, g_{leak,i},v_{leak,i}, w_{ij}, E_{ij}, \sigma_{ij}$, $\mu_{ij}, a_i$ and $b$ for $i,j \in \{1,\dots 8\}$ are learned during the training procedure.

\end{document}